\pdfoutput=1
\documentclass[acmsmall,nonacm,dvipsnames]{acmart}
\usepackage{microtype}
\usepackage{amsthm}
\usepackage{amsmath}
\usepackage{bm}
\usepackage{mathtools}
\usepackage{makecell}
\usepackage{threeparttable}
\usepackage{tabularx}
\usepackage{thmtools}
\usepackage{thm-restate}

\usepackage[noend]{algpseudocode}
\usepackage[many]{tcolorbox}
\newenvironment{algobox}[2][]{
    \begin{center}
    \begin{tcolorbox}[enhanced,breakable,title=\centering \large {#2},colback=white,colframe=black!80,width=\textwidth]
    \underline{\textbf{Code for a party $P_i${#1}}}
    \begin{algorithmic}[1]
}{
    \end{algorithmic} \end{tcolorbox} \end{center}
}
\algrenewcommand{\algorithmiccomment}[1]{\hfill \emph{\textcolor{red}{//{#1}}}}

\algblock{Upon}{EndUpon} \algrenewtext{Upon}[1]{\textbf{upon} {#1} \textbf{do}}
\makeatletter \ifthenelse{\equal{\ALG@noend}{t}} {\algtext*{EndUpon}} \makeatother
\algblock{When}{EndWhen} \algrenewtext{When}[1]{\textbf{when} {#1} \textbf{do}}
\makeatletter \ifthenelse{\equal{\ALG@noend}{t}} {\algtext*{EndWhen}} \makeatother

\DeclarePairedDelimiter{\floor}{\lfloor}{\rfloor}
\DeclarePairedDelimiter{\ceil}{\lceil}{\rceil}

\newcommand{\BO}{\mathcal{O}}
\newcommand{\sra}{\mathsf{SRA}}
\newcommand{\pra}{\mathsf{PRA}}

\newcommand{\ext}{\mathsf{EXT}}
\newcommand{\rec}{\mathsf{REC}}
\newcommand{\Enc}{\mathsf{Enc}}
\newcommand{\Dec}{\mathsf{Dec}}

\newcommand{\ba}{\mathsf{BA}}
\newcommand{\ca}{\mathsf{CA}}
\newcommand{\kca}{\mathsf{KCA}}

\newcommand{\KEY}{\texttt{KEY}}
\newcommand{\HASH}{\texttt{HASH}}
\newcommand{\SYM}{\texttt{SYM}}
\newcommand{\SUC}{\texttt{SUC}}
\newcommand{\YOURS}{\texttt{YOURS}}
\newcommand{\MINE}{\texttt{MINE}}
\newcommand{\msgpair}[2]{\langle {#1}, {#2} \rangle}

\begin{document}

\title{Extending Asynchronous Byzantine Agreement with Crusader Agreement}

\author{Mose Mizrahi Erbes} 
\orcid{0009-0009-9771-0845}
\email{mmizrahi@ethz.ch}
\affiliation{
  \institution{ETH Zurich}
  \country{Switzerland}
  \city{Zurich}
}

\author{Roger Wattenhofer}
\orcid{0000-0002-6339-3134}
\email{wattenhofer@ethz.ch}
\affiliation{
  \institution{ETH Zurich}
  \country{Switzerland}
  \city{Zurich}
}


\begin{abstract}
    In this work, we study multivalued byzantine agreement (BA) in an asynchronous network of $n$ parties where up to $t < \frac{n}{3}$ parties are byzantine. We present a new reduction from multivalued BA to binary BA. It allows one to achieve BA on $\ell$-bit inputs with one instance of binary BA, one instance of crusader agreement (CA) on $\ell$-bit inputs and $\Theta(\ell n + n^2)$ bits of additional communication. 
    
    As our reduction uses multivalued CA, we also design two new information-theoretic CA protocols for $\ell$-bit inputs. In the first one, we use almost-universal hashing to achieve statistical security with probability $1 - 2^{-\lambda}$ against $t < \frac{n}{3}$ faults with $\Theta(\ell n + n^2(\lambda + \log n))$ bits of communication. Following this, we replace the hashes with error correcting code symbols and add a preliminary step based on the synchronous multivalued BA protocol COOL [DISC '21] to obtain a second, perfectly secure CA protocol that can for any $\varepsilon > 0$ be set to tolerate $t \leq \frac{n}{3 + \varepsilon}$ faults with $\BO\bigl(\frac{\ell n}{\min(1, \varepsilon^2)} + n^2\max\bigl(1, \log \frac{1}{\varepsilon}\bigr) \bigr)$ bits of communication. Our CA protocols allow one to extend binary BA to multivalued BA with a constant round overhead, a quadratic-in-$n$ communication overhead, and information-theoretic security.
\end{abstract}

\maketitle

\section{Introduction}

In the byzantine agreement (BA) problem \cite{psl80, lsp82}, we have a network of $n$ parties with inputs, up to $t$ of which are byzantine. The parties wish to agree on a common output based on their inputs, such that if they have a common input, then they agree on it. For this, they run a protocol. The challenge is that the byzantine parties can arbitrarily deviate from the protocol.

On $\ell$-bit inputs, BA requires $\Omega(\ell n + \dots)$ bits of communication (when $t = \Omega(n)$) \cite{fh06}. Typically, BA protocols do not reach this bound, as reaching requires non-trivial techniques based on erasure/error correction that are not necessary for binary BA. Hence, a standard method to reach this bound is to take some other BA (or byzantine broadcast) protocol that does not as efficiently support long inputs, and running it with bit inputs or with short inputs (e.g.\ cryptographic hashes) in broader BA protocols. In the literature, these broader protocols are called \emph{extension protocols}. Extension protocols for BA (and byzantine broadcast) are widely studied \cite{tc84, fh06, pflr11, gp16, mr17, nayak20, gp21, chen21, lichen21, blln23, ck23, dm24, chen24, chen25, knr25}.

In this work, we consider an asynchronous network with up to $t < \frac{n}{3}$ byzantine parties. For this setting, we design a protocol $\ext$ that extends a binary BA protocol into a multivalued one. This extension is based on multivalued crusader agreement (CA) \cite{d82}, a weaker variant of BA where the parties agree on their common input if they have one, and otherwise for some value $v$ each output $v$ or $\bot$. Then, we design two efficient information-theoretic CA protocols, so that when we instantiate $\ext$ with them, we obtain the first information-theoretic asynchronous BA extensions we are aware of that tolerate $t \geq \frac{n}{5}$ faults with quadratic-in-$n$ overhead communication.

\subsection*{Contributions \& Related Work}

An important BA extension is Chen's COOL \cite{chen21}, which he recently refined in \cite{chen24}. COOL extends a binary BA protocol that optimally \cite{psl80,t84} tolerates $t < \frac{n}{3}$ faults with $\Theta(nt)$ bits of communication (e.g.\ \cite{bgp92}) into a BA protocol for $\ell$-bit inputs that tolerates $t < \frac{n}{3}$ faults with $\Theta(\ell n + nt)$ bits of communication. Note that the $\Theta(nt)$ term here is also optimal by the Dolev-Reischuk bound \cite{dr85}.

\begin{table}[h] \begin{threeparttable}
    \renewcommand{\tnote}[1]{\textsuperscript{#1}}
    \newcolumntype{Y}{>{\centering\arraybackslash}X}
    \setcellgapes{1.44pt}
    \setlength{\tabcolsep}{0.99\tabcolsep}
    \caption{The most communication-efficient asynchronous BA extensions we know. We assume $\lambda = \Omega(\log n)$ for the security parameter $\lambda$, assume $t = \Omega(n)$, let $C_A$ be the communication complexity of binary BA, and let $C_E$ be the communication complexity of electing a random leader party.}
    \makegapedcells
    \begin{tabularx}{1\columnwidth}{|c|c|c|Y|} \hline
        
        \textbf{Security} & \textbf{Threshold} & \textbf{Communication Complexity} & \textbf{References} \\ \hline

        perfect & $t < \frac{n}{3}$ & $\Theta(\ell n + nC_A + n^4\log n)$ & \cite{nayak20} \\ \hline
        
        perfect & $t < \frac{n}{3}$ & $\Theta(\ell n + nC_A + n^3\log n)$\tnote{a} & \cite{chen25} \\ \hline

        perfect & $t < \frac{n}{3}$ & $\Theta(\ell n + C_E + C_A + n^3\log n)$\tnote{b} & \cite{chen25} \\ \hline

        perfect & $t < \frac{n}{5}$ & $\Theta(\ell n + C_A + n^2)$\tnote{c} & \cite{lichen21} \\ \hline

        \textcolor{blue}{perfect} & \textcolor{blue}{$t \leq \frac{n}{3 + \varepsilon}$} & \textcolor{blue}{$\BO\bigl(\frac{\ell n}{\min(1, \varepsilon^2)} + C_A + n^2\max\bigl(1, \log \frac{1}{\varepsilon}\bigr) \bigr)$} & \textcolor{blue}{this paper ($\ext + \ca_2$)} \\ \hline

        \textcolor{blue}{statistical} & \textcolor{blue}{$t < \frac{n}{3}$} & \textcolor{blue}{$\Theta(\ell n + C_A + n^2\lambda)$} & \textcolor{blue}{this paper ($\ext + \ca_1$)} \\ \hline

        cryptographic & $t < \frac{n}{3}$ & $\Theta(\ell n + C_A + n^2\lambda)$ & \cite{nayak20} \\ \hline

        \makecell{cryptographic, \\ no front-running\tnote{d}} & $t \leq \frac{n}{3 + \varepsilon}$ & \makecell{$\Theta(\ell n + n \cdot \mathsf{poly}(\lambda))$, \\ assuming $\varepsilon > 0$ is constant} & \cite{blln23, ck23} \\ \hline
    
    \end{tabularx}
    \begin{tablenotes}
        \item[a, b] The $\log n$ factor in the term $n^3\log n$ is due to the $\Omega(\log n)$-bit protocol ID tags necessary to distinguish $n$ reliable broadcasts, which \cite{chen25} does not count toward complexity.
        \item[c] In \cite{lichen21}, the complexity when $t = \Omega(n)$ is $\BO(\ell n + C_A + n^2\log n)$. However, one can drop the $\log n$ factor by using near-MDS error correcting codes \cite{gi04, rs06, lm25} instead of Reed-Solomon codes.
        \item[d] No front-running means that the adversary cannot drop a message that a party sends by corrupting the party after the party sends the message. This restriction allows \cite{blln23,ck23} to circumvent the $\Omega(t^2)$ message complexity lower bound in \cite{acdnprs22} for adaptively secure BA.
    \end{tablenotes}
    \label{comparisons}
\end{threeparttable} \end{table}

An ``ideal'' asynchronous BA extension would match COOL. Such an extension protocol would \begin{itemize}
    \item run an underlying binary BA protocol only once,
    \item have a constant round overhead and a $\BO(\ell n + nt)$ communication complexity overhead,\footnote{The overhead of an extension is the extension's complexity excluding the execution(s) of the underlying binary BA protocol.}
    \item tolerate $t < \frac{n}{3}$ faults with perfect security.
\end{itemize}

Though we do not reach this ideal, we progress towards it with an asynchronous BA extension $\ext$ based on asynchronous data dissemination (ADD) \cite{dxr21} and on CA. Against $t < \frac{n}{3}$ faults, $\ext$ reduces BA on $\ell$-bit inputs to one instance of binary BA, one instance of CA on $\ell$-bit inputs, a constant number of additional rounds, $\Theta(n^2)$ extra messages and $\Theta(\ell n + n^2)$ bits of extra communication.

We follow $\ext$ with two information-theoretic CA protocols: a statistically secure protocol $\ca_1$ and a perfectly secure protocol $\ca_2$. The first one achieves security with probability $1 - 2^{-\lambda}$ against $t < \frac{n}{3}$ faults with $\Theta(\ell n + n^2(\lambda + \log n))$ bits of communication, while the second one can for any $\varepsilon > 0$ be set to tolerate $t \leq \frac{n}{3 + \varepsilon}$ faults with $\BO\bigl(\frac{\ell n}{\min(1, \varepsilon^2)} + n^2\max\bigl(1, \log \frac{1}{\varepsilon}\bigr)\bigr)$ bits of communication. These CA protocols are why we call our extensions $\ext + \ca_1$ and $\ext + \ca_2$ in Table \ref{comparisons}: These are $\ext$ instantiated with $\ca_1$ and $\ca_2$. Our CA protocols, which to our knowledge are not matched in the literature and are thus also of independent interest, allow us to obtain information-theoretically secure asynchronous BA extensions that run an underlying binary BA protocol only once, with a constant round overhead and a quadratic-in-$n$ communication overhead.

Among perfectly secure asynchronous BA extensions, our $\ext + \ca_2$ improves upon the A-COOL (asynchronous COOL) protocol by Li and Chen \cite{lichen21}. While A-COOL achieves the communication complexity $\Theta(\ell n + C_A + nt)$ against $t < \frac{n}{5}$ faults, our $\ext + \ca_2$ achieves the communication complexity $\BO\bigl(\frac{\ell n}{\min(1, \varepsilon^2)} + C_A + n^2\max\bigl(1, \log \frac{1}{\varepsilon}\bigr) \bigr)$ against $t \leq \frac{n}{3 + \varepsilon}$ faults. Whenever $\varepsilon > 0$ is constant, $\BO\bigl(\frac{\ell n}{\min(1, \varepsilon^2)} + C_A + n^2\max\bigl(1, \log \frac{1}{\varepsilon}\bigr) \bigr) = \BO(\ell n + C_A + n^2)$, which means that our $\ext + \ca_2$ can match A-COOL in terms of asymptotic complexity while providing greater fault tolerance. 

Meanwhile, our statistically secure $\ext + \ca_1$ is the most communication-efficient information-theoretic asynchronous BA extension we know of for $t \geq \frac{n}{5}$ and $\ell = o(n^2)$. In terms of asymptotic communication complexity, it matches the state-of-the-art cryptographic extension by Nayak, Ren, Shi, Vaidya and Xiang \cite{nayak20}, and it is only surpassed by the cryptographic extensions in \cite{blln23, ck23}, which tolerate $t \leq \frac{n}{3 + \varepsilon}$ faults, and limit the adversary's adaptivity to circumvent a $\Omega(t^2)$ message complexity lower bound for adaptively secure BA \cite{acdnprs22}.

A binary BA protocol that pairs well with our extensions is the statistically secure one in \cite{kkkss10} that tolerates $t$ static faults (which occur before protocol execution) with $\widetilde{\Theta}(n^2)$ bits of communication if $t \leq \frac{n}{3 + \varepsilon}$ for any constant $\varepsilon > 0$. Another one is the perfectly secure binary BA protocol in \cite{mmr15}, which assumes that the parties can access an ideal common coin (with which they can repeatedly obtain common random bits), and uses this assumption to tolerate $t < \frac{n}{3}$ adaptive corruptions with $\Theta(n^2)$ bits of expected communication in $\Theta(1)$ expected rounds.

Some other notable perfectly secure BA extensions are the ones in Chen's concurrent work \cite{chen25}, where Chen achieves multivalued asynchronous BA with either $n$ binary BA instances or with $\Theta(1)$ expected leader elections and binary BA instances, all coupled with $\Theta(\ell n + n^3\log n)$ bits of overhead communication. These protocols tolerate $t < \frac{n}{3}$ faults, which means that they asymptotically have less overhead communication than $\ext + \ca_2$ if $t \approx \frac{n}{3}$, but more overhead if $t \leq \frac{n}{3 + \varepsilon}$ for any constant $\varepsilon > 0$. Before these protocols, the most efficient perfectly secure extension protocol for $t < \frac{n}{3}$ faults was the one in \cite{nayak20} by Nayak, Ren, Shi, Vaidya and Xiang, which has $\Theta(\ell n + n^4\log n)$ bits of overhead communication, and runs a binary BA protocol $2n$ times. This extension follows the the design of an earlier one by \mbox{Patra with $\Theta(\ell n + n^5\log n)$ bits of overhead communication \cite{pflr11}.}

\subsection*{Overview}

Our extensions begin with a common starting point: An ``incomplete'' extension $\ext$. This extension is based on asynchronous data dissemination \cite{dxr21} and on multivalued CA. Against $t < \frac{n}{3}$ byzantine faults, $\ext$ reduces BA on $\ell$-bit inputs to one instance of CA on $\ell$-bit inputs, one instance of binary BA, a constant number of additional rounds and $\Theta(\ell n + n^2)$ bits of additional communication. So, $\ext$ lets us convert efficient multivalued CA protocols into efficient BA extensions. Note that there exist comparable extension protocols that use \mbox{graded consensus instead of crusader agreement \cite{knr25}.}

Following $\ext$, which we present in Section \ref{reducsec}, we design two new efficient information-theoretic crusader agreement protocols ($\ca_1$ and $\ca_2$) for use in $\ext$, respectively in Section \ref{statsec} and \mbox{Section \ref{perfsec}.} 

In $\ca_1$, we use statistical equality checks based on almost-universal hashing \cite{cw79}, which Fitzi and Hirt used to obtain the first BA extension in the literature with the complexity $\BO(\ell n + \dots)$ \cite{fh06}. Almost-universal hashing lets two parties with the values $a$ and $b$ agree on a joint random key $k$, and then use a keyed hash function $h(k, \cdot)$ to check if $a = b$ by checking if $h(k, a) = h(k, b)$, with little communication and only a negligible hash collision probability for distinct values. Using almost-universal hashing, we achieve crusader agreement with $\Theta(\ell n + n^2(\lambda + \log n))$ bits of communication except with probability less than $2^{-\lambda}$ against $t < \frac{n}{3}$ faults.

Then, we obtain $\ca_2$ by trying to upgrade $\ca_1$ into a perfectly secure protocol. We keep the overall design of $\ca_1$, which we present in Section \ref{statsec}. However, in $\ca_2$, the parties do not compare keyed hashes. Instead, they encode values with error correcting codes (ECC), and compare indexed ECC symbols. Unlike hash comparisons, these symbol comparisons are error-prone, as distinct strings can have partially coinciding ECC encodings. To limit these errors to a tolerable amount (to less than $\frac{n - 3t}{2} \geq \frac{\varepsilon n}{6 + 2\varepsilon}$ errors per party when $t \leq \frac{n}{3 + \varepsilon}$ for some $\varepsilon > 0$), we use $(n, \delta)$-error correcting codes with $\frac{n}{\delta} = \BO(\min(1, \varepsilon)^{-2})$, and we introduce a preliminary step based on the first two phases of Chen's COOL protocol \cite{chen21}, where we decrease the maximum possible number of non-$\bot$ values the parties might want to agree on from $n$ to at most $\ceil[\big]{\frac{8}{\min(1, \varepsilon)}}$.

\section{Preliminaries}

We consider an asynchronous network of $n$ message-passing parties $P_1, \dots, P_n$, fully connected via reliable and authenticated channels. We design our protocols against an adaptive adversary that corrupts up to $t$ parties during a protocol's execution, depending on the information it gathers as the protocol runs. A corrupted party becomes byzantine, and starts behaving arbitrarily. Meanwhile, a party who is never corrupted is honest. The adversary must deliver all messages from honest senders, but it can otherwise schedule messages as it wishes. It can also adaptively ``front-run'' messages, i.e.\ drop any message by corrupting the message's previously non-corrupt sender instead of delivering the message. Note that some works (e.g.\ the BA protocols/extensions \cite{bkll20,cks20,blln23,ck23}) forbid this to circumvent a lower bound of $\Omega(t^2)$ messages for adaptively secure BA \cite{acdnprs22}.

We work in the \emph{full-information} model where the adversary can see all messages and all parties' internal states. However, the adversary cannot predict the random values a party will generate.

When a party sends a message $m$ to every party, it \emph{multicasts} $m$. The adversary can corrupt a party while the party is multicasting a message $m$, and then deliver $m$ to only some parties.

The parties do not begin protocols with inputs. Rather, a party can \emph{acquire} an input while it is running a protocol. This will matter for $\ext$, which will terminate even if some parties never input to the underlying binary BA protocol, and for $\ca_1$, which will achieve statistical security against \emph{adversarially chosen inputs} (i.e.\ security when the adversary can determine the parties' $\ca_1$ inputs).

\paragraph*{Error Correcting Codes.} Like previous papers that aim for the complexity $\BO(\ell n + \dots)$, we use $(n, k)$-error correcting codes, which allow us to encode a bitstring $m$ of length $\ell$ into $n$ symbols in an alphabet $\Sigma = \{0,1\}^a$ where $a \approx \frac{\ell}{k}$, and then recover $m$ with some missing or incorrect symbols. Formally, we use the following encoding and decoding functions: \begin{itemize}
    \item $\Enc_k(m)$: This encoding function takes a string $m \in \{0, 1\}^\ell$, and encodes it using an $(n, k)$-code into a list of $n$ symbols $(s_1, \dots, s_n)$, each in $\Sigma$. 
    \item $\Dec_k(s_1, \dots, s_n)$: This decoding function decodes a list of $n$ symbols $s_1, \dots, s_n \in \Sigma \cup \{\bot\}$ (where $\bot$ denotes missing symbols) into a string $m \in \{0, 1\}^\ell$. With respect to any string $m'$ where $\Enc_k(m') = (s_1', \dots, s_n')$, we call each symbol $s_i'$ correct if $s_i' = s_i$, missing if $s_i' = \bot$, and incorrect otherwise. If with respect to a string $m' \in \{0, 1\}^\ell$ the number of incorrect symbols $c$ and the number of missing symbols $d$ obey $2c + d \leq n - k$, then $\Dec_k(s_1, \dots, s_n)$ returns $m'$. Else, $\Dec_k(s_1, \dots, s_n)$ can return any $\ell$-bit string or the failure indicator $\bot$.
\end{itemize}

In the multivalued agreement literature, the usual error correcting codes are Reed-Solomon codes \cite{reed-solomon}. These codes permit $a = \Theta\bigl(\frac{\ell}{k} + \log n\bigr)$, as they require $a > n$ So, one usually finds a logarithmic communication complexity factor in works that use them. For instance, COOL as it is written requires $\BO(\ell n + nt\log t)$ bits of communication \cite{chen21}. However, Chen has recently proposed a way to drop the $\log n$ factor \cite{chen2024MVBA,chen25}. Instead of Reed-Solomon $(n, k)$-codes, one can use near-MDS $(n, k)$-codes \cite{gi04, rs06, lm25}, which permit smaller symbols. Concretely, there are near-MDS codes that permit $a = \Theta\bigl(\frac{\ell}{k} + \log \frac{n}{k}\bigr)$ \cite{lm25}. With this, $k = n - 2t > \frac{n}{3}$ (as in $\ext$ and $\ca_1$) allows symbols of size $\Theta\bigl(\frac{\ell}{n} + 1\bigr)$, while $k = \Omega(\min(1,\varepsilon^2) n)$ (as in $\ca_2$) allows symbols of size $\BO\bigl(\frac{\ell}{n \cdot \min(1, \varepsilon^2)} + \max\bigl(1, \log \frac{1}{\varepsilon}\bigr)\bigr)$.

\enlargethispage{0.9pt}

\subsection*{Protocols}

Below, we provide property-based definitions for the 5 kinds of protocols we use, with the implicit assumption that the honest parties run forever. Otherwise, some parties could fail to output even if a protocol required them to. However, we will design some of our protocols (e.g. $\ext$) so that they \emph{terminate}, i.e.\ so that each party can safely terminate (stop running) a protocol when it outputs.

\paragraph*{Byzantine Agreement.} In byzantine agreement, each party $P_i$ inputs some $v_i \in \{0, 1\}^\ell$ and outputs some $y_i \in \{0, 1\}^\ell \cup \{\bot\}$, where $\ell$ is some publicly known length parameter. A BA protocol achieves: \begin{itemize}
    \item \textbf{Validity:} If some $v \in \{0,1\}^\ell$ is the only honest input, then every honest output is $v$.
    \item \textbf{Consistency:} For some $y \in \{0, 1\}^\ell \cup \{\bot\}$, every honest output is $y$.
    \item \textbf{Liveness:} If every honest party acquires an input, then every honest party outputs.
    \item \textbf{Intrusion Tolerance:} Every honest output is either an honest input or $\bot$.
\end{itemize} The fundamental properties for BA are validity, consistency and liveness. Meanwhile, intrusion tolerance \cite{mr10} is a desirable property for multivalued BA that $\ext$ achieves for free. Intrusion tolerance is the reason why we have the output $\bot$: When the parties fail to agree on an honest input, they agree on the safe fallback output $\bot$. We do not allow the output $\bot$ for binary BA, because for binary inputs, validity implies intrusion tolerance even if the output $\bot$ is mapped to the output $0$.

\paragraph{Crusader Agreement.} In crusader agreement, each party $P_i$ inputs some $v_i \in \{0, 1\}^\ell$, and outputs some $y_i \in \{0, 1\}^\ell \cup \{\bot\}$. A crusader agreement protocol achieves: \begin{itemize}
    \item \textbf{Validity:} If some $v \in \{0,1\}^\ell$ is the only honest input, then every honest output is $v$.
    \item \textbf{Weak Consistency:} For some $y \in \{0,1\}^\ell$, each honest output is either $y$ or $\bot$.
    \item \textbf{Liveness:} If every honest party acquires an input, then every honest party outputs.
\end{itemize}

\paragraph*{$k$-Crusader Agreement.} A $k$-crusader agreement protocol is a crusader agreement protocol that instead of weak consistency achieves the following parametrized weakening of weak consistency: \begin{itemize}
    \item $\bm{k}$-\textbf{Weak Consistency:} For some $y_1, \dots, y_k \in \{0,1\}^\ell$, each honest output is in $\{y_1, \dots, y_k, \bot\}$.
\end{itemize}
In COOL \cite{chen21}, Chen uses $2$-crusader agreement as a stepping stone for crusader agreement and then BA in a synchronous network. Likewise, we will use $\ceil[\big]{\frac{8}{\min(1, \varepsilon)}}$-crusader agreement in $\ca_2$.

\paragraph*{Reliable Agreement.} In reliable agreement \cite{chl21, ddlmrs24}, each party $P_i$ inputs some $v_i \in \{0, 1\}^\ell$, and possibly outputs some $y_i \in \{0, 1\}^\ell$. A reliable agreement protocol achieves: \begin{itemize}
    \item \textbf{Validity:} If some $v \in \{0,1\}^\ell$ is the only honest input, then every honest output is $v$.
    \item \textbf{Consistency:} For some $y \in \{0, 1\}^\ell \cup \{\bot\}$, every honest output is $y$.
    \item \textbf{Liveness:} If the honest parties all acquire a common input, they they all output.
\end{itemize}
Unlike BA, reliable agreement only requires the parties to output if they all have a common input. This makes it easier than even crusader agreement. While \cite{ddlmrs24} uses reliable agreement (with some more properties) for termination, we will use it to ensure \mbox{weak consistency in $\ca_1$ and $\ca_2$.}

\paragraph*{Reconstruction.} Sometimes, the parties must all learn a common value $v^*$ that only (at least) $t + 1$ of them initially know. For this, we define reconstruction, where each party $P_i$ possibly inputs some $v_i \in \{0, 1\}^\ell$, and possibly outputs some $y_i \in \{0, 1\}^\ell$. To work properly, a reconstruction protocol requires there to exist some $v^* \in \{0, 1\}^\ell$ such that no honest party acquires any other input. When this is the case, a reconstruction protocol achieves: \begin{itemize}
    \item \textbf{Validity:} The honest parties can only output $v^*$, and they can do so only after some honest party acquires the input $v^*$.
    \item \textbf{Liveness:} If at least $t + 1$ honest parties acquire the input $v^*$, then every honest party \nolinebreak outputs.
    \item \textbf{Totality:} If some honest party outputs, then every honest party outputs.
\end{itemize}
Roughly, reconstruction is asynchronous data dissemination (ADD) \cite{dxr21}. The two minor differences are that ADD does not require totality, and that while in ADD a party without a bitstring input acquires the input $\bot$, in reconstruction such a party simply never acquires an input. In the \hyperref[recsec]{appendix}, we present and explain a perfectly secure terminating reconstruction protocol $\rec$ based on online error correction which we obtain by slightly modifying the ADD protocol in \cite{dxr21}. The complexity of $\rec$ is a constant number of rounds, $\Theta(n^2)$ messages and $\Theta(\ell n + n^2)$ bits of communication, like the ADD protocol in \cite{dxr21} (when it is implemented with an error correcting code with a constant-size alphabet). We will use $\rec$ once in each of $\ext$, $\ca_1$ and $\ca_2$, though only $\ext$ will benefit from the totality of $\rec$ and the fact that $\rec$ terminates without requiring the parties to run it forever.

\section{Extension via Crusader Agreement} \label{reducsec}

In this section, we present the extension $\ext$, where we use the terminating reconstruction protocol $\rec$, any terminating binary byzantine agreement protocol $\Pi_\ba$, and any intrusion tolerant crusader agreement\footnote{One can add intrusion tolerance (the guarantee that the honest parties output honest inputs or $\bot$) to any crusader agreement protocol $\Pi_\ca$ with a final step where each party checks if its $\Pi_\ca$ output equals its $\Pi_\ca$ input, and outputs $\bot$ if not.} protocol $\Pi_\ca$ for $\ell$-bit inputs. The protocol $\Pi_\ba$ must achieve the totality guarantee that if some honest party terminates it, then every honest party terminates it, even if not every honest party acquires an input. Note that one can easily add termination and this guarantee to any binary BA protocol with the terminating reliable agreement protocol in \cite{ddlmrs24}, which has the same guarantee. This reliable agreement protocol takes a constant number of rounds, and it involves $\BO(n^2)$ bits of communication just like Bracha's reliable broadcast protocol \cite{bracha87} which it is based on.

\begin{algobox}[]{Protocol $\ext$}
    \State let $y_i \gets \bot$
    \State join common instances of $\rec$, $\Pi_\ca$, and $\Pi_\ba$

    \Upon{acquiring the input $v_i$}
        \State input $v_i$ to $\Pi_\ca$
    \EndUpon

    \Upon{outputting some $v^* \in \{0, 1\}^\ell$ from $\Pi_\ca$}
        \State input $v^*$ to $\rec$
    \EndUpon

    \Upon{outputting $\bot$ from $\Pi_\ca$}
        \State multicast $\bot$
        \State if you have not input to $\Pi_\ba$ before, input 0 to $\Pi_\ba$
    \EndUpon

    \Upon{receiving $\bot$ from $t + 1$ parties}
        \State if you have not input to $\Pi_\ba$ before, input 0 to $\Pi_\ba$
    \EndUpon

    \Upon{terminating $\Pi_\ba$ with the input $0$}
        \State output $\bot$ and terminate
    \EndUpon

    \Upon{terminating $\rec$ with the output $v^*$}
        \State $y_i \gets v^*$
        \State if you have not input to $\Pi_\ba$ before, input 1 to $\Pi_\ba$
    \EndUpon

    \Upon{terminating $\Pi_\ba$ with the output 1}
        \State when $y_i \neq \bot$, output $y_i$ and terminate
    \EndUpon
\end{algobox}

Let us discuss the cryptographic extension in \cite{nayak20}. In it, the parties encode their inputs with an $(n, n-2t)$-error correcting code, and cryptographically accumulate their indexed symbols. Then, they achieve BA (with intrusion tolerance) on their accumulation inputs, agreeing on either $\bot$ or the accumulation of a party with some input $v^*$. In the former case, the extension outputs $\bot$. In the latter case, the extension outputs $v^*$ with the help of the party (or parties) with the input $v^*$ who can prove that they are sending the correct symbols w.r.t.\ $v^*$ thanks to the accumulation.

In $\ext$, we similarly decide with $\Pi_\ba$ if the parties should agree on some ``correct'' $v^*$ or on $\bot$. However, our $\Pi_\ba$ is binary, and we establish $v^*$ not with $\Pi_\ba$ but with $\Pi_\ca$. The input determination for $\Pi_\ba$ is not just ``input 1 to $\Pi_\ba$ iff you output $v^* \neq \bot$ from $\Pi_\ca$,'' as then $\Pi_\ba$ could output $1$ with just one honest party who knows $v^*$, and this party would be unable to convince the others to output $v^*$ due to the lack of an accumulation. Instead, a party who outputs $v^* \neq \bot$ from $\Pi_\ca$ inputs $v^*$ to $\rec$, and a party inputs 1 to $\Pi_\ba$ only if it terminates $\rec$. This way, if a party outputs 1 from $\Pi_\ba$, then everyone terminates $\rec$ and learns $v^*$ thanks to $\rec$'s totality. Note that $\rec$ is able to achieve totality because in it an honest party can only terminate after $t + 1$ honest parties know $v^*$.

Below, we prove $\ext$ secure. While doing so, we refer to some string $v^* \in \{0,1\}^\ell$ such that the honest parties output only $v^*$ or $\bot$ from $\Pi_\ca$. The bitstring $v^*$ is unique if some honest party outputs a bitstring from $\Pi_\ca$, and it can be chosen arbitrarily if the honest parties output only $\bot$ from $\Pi_\ca$. 

\begin{lemma} \label{onlyvstarlemma}
    In $\ext$, the string $v^*$ is the only $\rec$ input the honest parties can acquire.
\end{lemma}

\begin{proof}
    If an honest party inputs some $v \neq \bot$ to $\rec$, then it has output $v$ from $\Pi_\ca$, and therefore $v = v^*$ since $v^*$ is the only non-$\bot$ output any honest party can obtain from $\Pi_\ca$.
\end{proof}

We use Lemma \ref{onlyvstarlemma} in the rest of this section without explicitly referring to it. 

\begin{lemma} \label{extsomelemma}
    In $\ext$, some honest party terminates $\Pi_\ba$.
\end{lemma}

\begin{proof}
    For the sake of contradiction, consider an execution of $\ext$ where Lemma \ref{extsomelemma} fails. In this execution, the honest parties never terminate $\Pi_\ba$, and thus never terminate $\ext$. So, they all acquire inputs, run $\Pi_\ca$ with their inputs, and output from $\Pi_\ca$. As there are at least $2t + 1$ honest parties, either $t + 1$ honest parties output $v^*$ from $\Pi_\ca$, or $t + 1$ honest parties output $\bot$ from $\Pi_\ca$.
     \begin{itemize}
        \item If $t + 1$ honest parties output $v^*$ from $\Pi_\ca$, then because these parties input $v^*$ to $\rec$, the honest parties all terminate $\rec$. Hence, the honest parties can all input 1 to $\Pi_\ba$.
        \item If $t + 1$ honest parties output $\bot$ from $\Pi_\ca$, then because these parties multicast $\bot$, the honest parties all receive $\bot$ from $t + 1$ parties. Hence, the honest parties can all input 0 to $\Pi_\ba$.
    \end{itemize}
    No matter what, the honest parties all input to $\Pi_\ba$. So, contradictorily, they all terminate $\Pi_\ba$.
\end{proof}

\begin{theorem}
    The protocol $\ext$ achieves validity, consistency, intrusion tolerance and termination.
\end{theorem}

\begin{proof} First, we prove that $\ext$ terminates with either $v^*$ or $\bot$ as the common output. By Lemma \ref{extsomelemma}, some honest party terminates $\Pi_\ba$. So, by the totality guarantee we have assumed for $\Pi_\ba$, the honest parties all terminate $\Pi_\ba$, even if they do not all acquire $\Pi_\ba$ inputs. If $\Pi_\ba$ outputs $0$, then the parties terminate $\ext$ with the output $\bot$. Meanwhile, if $\Pi_\ba$ outputs $1$, then some honest party must have input $1$ to $\Pi_\ba$ upon terminating $\rec$. So, by the validity and totality of $\rec$, the honest parties all terminate $\rec$ with the output $v^*$, and thus they all terminate $\ext$ with the output $v^*$.

Next, we prove validity. Suppose the honest parties run $\ext$ with a common input. Then, this input is the unique output they can obtain from $\Pi_\ca$; i.e.\ it is $v^*$. Since the honest parties do not output $\bot$ from $\Pi_\ca$, they do not multicast $\bot$, they do not receive $\bot$ from $t + 1$ parties, and they do not input $0$ to $\Pi_\ba$. This means that they can only output $1$ from $\Pi_\ba$ and output $v^*$ from $\ext$.
       
Finally, we prove intrusion tolerance. If some honest party outputs $v^*$ from $\Pi_\ca$, then $v^*$ is an honest input as $\Pi_\ca$ is intrusion tolerant. Hence, $\ext$ achieves intrusion tolerance by outputting either $v^*$ or $\bot$. Meanwhile, if the honest parties only output $\bot$ from $\Pi_\ca$, then they do not input to $\rec$. Hence, they cannot terminate $\rec$, which means that they can only output $\bot$ from $\ext$. \qedhere
\end{proof}

\paragraph*{Complexity of \textsf{EXT}} The round/message/communication complexities of $\ext$ are those of $\Pi_\ca$ and $\Pi_\ba$ added, plus (due to $\rec$ and the $\bot$ messages in $\ext$) a constant number of rounds, $\BO(n^2)$ messages and $\BO(\ell n + n^2)$ bits of communication.

\section{Statistical Security Against \texorpdfstring{$t < \frac{n}{3}$}{t < n/3} Corruptions} \label{statsec}

Before we continue, let us define the main tool we will use for statistical security.

\paragraph*{Almost Universal Hashing.} Let $h(k, m) : \{0,1\}^\kappa \times \{0,1\}^\ell \rightarrow \{0,1\}^\kappa$ be a keyed hash function that outputs a $\kappa$-bit hash for each $\kappa$-bit key $k$ and $\ell$-bit message $m$. Based on the definition in \cite{fh06}, we say that $h$ is $\varepsilon$-almost universal if $\Pr[h(k, a) = h(k, b)] \leq \varepsilon$ for any distinct $a, b \in \{0,1\}^\ell$ when the key $k$ is chosen uniformly at random from $\{0,1\}^\kappa$. In this section, we use a $\big(\frac{2^{-\lambda}}{n^2}\big)$-almost universal keyed hash function $h$ so that after less than $n^2$ hash comparisons, the probability that any collision has occurred is by the union bound less than $2^{-\lambda}$. We obtain $h$ by setting $\kappa = \ceil{\lambda + \log_2(\ell n^2) + 1}$ and using polynomials in $GF(2^\kappa)$. Further details on $h$ can be found in the \hyperref[almostunivappendix]{appendix}.

\subsection{Statistically Secure Reliable Agreement}

Let us warm up with the simple reliable agreement protocol $\sra$ below. In $\sra$, each party $P_i$ with the input $v_i$ chooses and multicasts a random key $k_i \in \{0, \dots, 2^{\kappa}-1\}$, establishes the joint key $k_{\{i,j\}} = k_i + k_j \bmod 2^\kappa$ with each party $P_j$ that multicasts a key $k_j$, and sends $h(k_{\{i,j\}}, v_i)$ to $P_j$ so that $P_j$ can check if $v_i$ equals its own input $v_j$ by checking if $h(k_{\{i,j\}}, v_i) = h(k_{\{i,j\}}, v_j)$.

For \emph{adversarial input tolerance} when the adversary can decide the parties' inputs while $\sra$ runs (e.g.\ by influencing a preceding protocol), the joint key $k_{\{i, j\}}$ of any honest $P_i$ and $P_j$ should be uniformly random in $\{0, \dots, 2^\kappa - 1\}$ \emph{given the parties' inputs} $(v_i, v_j)$. For this, we only let parties choose keys after they acquire inputs. WLOG, if $P_j$ acquires an input after $P_i$, then the key $k_{\{i, j\}}$ is uniformly random given $(v_i, v_j)$ as it is randomized by $k_j$, which is independent from $v_i$, $v_j$ and $k_i$.

\begin{algobox}[ who eventually acquires the input $v_i$]{Protocol $\sra$}
    \State \textit{\textcolor{red}{// Note: The party $P_i$ must wait until it acquires $v_i$ before running the code below}}
    \State let $A_i \gets \{P_i\}$, and let $(k_{\{i, 1\}}, \dots, k_{\{i, n\}}) \gets (\bot, \dots, \bot)$
    \State let $k_i$ be a uniformly random key in $\{0, \dots, 2^{\kappa} - 1\}$, and multicast $\msgpair{\KEY}{k_i}$

    \Upon{receiving some $\msgpair{\KEY}{k_j}$ from any other party $P_j$ for the first time}
        \State $k_{\{i, j\}} \gets (k_i + k_j) \bmod 2^\kappa$
        \State send $\msgpair{\HASH}{h(k_{\{i, j\}}, v_i)}$ to $P_j$
    \EndUpon

    \Upon{receiving some $\msgpair{\HASH}{z}$ from any other party $P_j$ for the first time, as soon as $k_{\{i, j\}} \neq \bot$}
        \State if $z = h(k_{\{i, j\}}, v_i)$, then add $P_j$ to $A_i$
        \State if $|A_i| = n - t$, then output $v_i$, and keep running
    \EndUpon
\end{algobox}

Other than the subtlety of key/input independence, the protocol $\sra$ is simple: A party waits (potentially forever) until it observes $n - t$ matching hashes, and outputs its input if it does so.

\begin{theorem} \label{srasec}
    Except with probability less than $2^{-\lambda - 1}$, $\sra$ achieves validity and consistency.
\end{theorem}

\begin{proof}
    Each honest $P_i$ initializes a set of parties $A_i$ to keep track of the parties with the same input \linebreak as itself, initially including only itself in $A_i$. Then, $P_i$ generates a random key $k_i$, multicasts $\msgpair{\KEY}{k_i}$, and agrees on the joint random key $k_{\{i, j\}} = (k_i + k_j) \bmod 2^\kappa$ with each $P_j$ from which it receives $\msgpair{\KEY}{k_j}$, so that if $P_j$ is honest, then $k_{\{i, j\}}$ is independent from $(v_i, v_j)$ due to being randomized by both $P_i$ and $P_j$. When $P_i$ knows $k_{\{i, j\}}$, it sends $\msgpair{\HASH}{h(k_{\{i, j\}}, v_i)}$ (the hash of $v_i$ w.r.t.\ the key $k_{\{i, j\}}$) to $P_j$. Finally, $P_i$ adds $P_j$ to $A_i$ if from $P_j$ it receives the matching hash $\msgpair{\HASH}{h(k_{\{i, j\}}, v_i)}$. \begin{itemize}
        \item \textit{Validity:} If the honest parties have a common input $v^*$, then they send each other matching hashes. So, they all observe $n - t$ matching hashes, and output the common input $v^*$.
        \item \textit{Consistency:} Suppose two honest parties $P_i$ and $P_j$ respectively output their inputs $v_i$ and $v_j$. Then, when $P_i$ and $P_j$ have both output, $|A_i \cap A_j| \geq n - 2t \geq t + 1$ since $|A_i| \geq n - t$ and $|A_j| \geq n - t$, and thus $A_i \cap A_j$ contains an honest party $P_q$. From this, we get that $P_q$'s input $v_q$ equals $v_i$ unless $v_i$ and $v_q$ collide on the key $k_{\{i, q\}}$ (i.e.\ $v_i \neq v_q$ but $h(k_{\{i, q\}}, v_i) = h(k_{\{i, q\}}, v_q)$), and that $v_q = v_j$ unless $v_q$ and $v_j$ collide on the key $k_{\{q, j\}}$. Assuming no hash collisions, we have $v_i = v_q = v_j$. This means that consistency is achieved unless there exist two distinct honest parties $P_i$ and $P_j$ whose inputs $v_i$ and $v_j$ collide on the key $k_{\{i, j\}}$. The probability that there exists such a pair of parties is at most $\binom{n}{2} \cdot \frac{2^{-\lambda}}{n^2} < 2^{-\lambda - 1}$, by the union bound. \qedhere
    \end{itemize}
\end{proof}

\paragraph*{Complexity of \textsf{SRA}} The protocol $\sra$ consists of 2 rounds: one to establish joint keys, and one to exchange input hashes. Its message complexity is $\Theta(n^2)$, and its communication complexity is $\Theta(n^2\kappa) = \Theta(n^2(\lambda + \log (\ell n)))$ since each message carries either a $\kappa$-bit key or a $\kappa$-bit hash.

\subsection{Statistically Secure Crusader Agreement}

We are now ready for the more challenging $\ca_1$, in which we build upon the key and hash exchange mechanism of $\sra$, and use both $\sra$ and the reconstruction protocol $\rec$ as subprotocols.

\begin{algobox}[ who eventually acquires the input $v_i$]{Protocol $\ca_1$}
    \State join an instance of $\rec$ and an instance of $\sra$
    \State \textit{\textcolor{red}{// Note: The party $P_i$ must wait until it acquires $v_i$ before running the code below}}
    \State let $A_i, B_i, C_i \gets \{P_i\}, \{\}, \{\}$, and let $(k_{\{i, 1\}}, \dots, k_{\{i, n\}}) \gets (\bot, \dots, \bot)$
    \State let $k_i$ be a uniformly random key in $\{0, \dots, 2^{\kappa} - 1\}$, and multicast $\msgpair{\KEY}{k_i}$
    
    \Upon{receiving some $\msgpair{\KEY}{k_j}$ from any other party $P_j$ for the first time}
        \State $k_{\{i, j\}} \gets (k_i + k_j) \bmod 2^\kappa$
        \State send $\msgpair{\HASH}{h(k_{\{i, j\}}, v_i)}$ to $P_j$
    \EndUpon

    \Upon{receiving some $\msgpair{\HASH}{z}$ from any other party $P_j$ for the first time, as soon as $k_{\{i, j\}} \neq \bot$}
        \If{$z = h(k_{\{i, j\}}, v_i)$}
            \State add $P_j$ to $A_i$
            \State if $|A_i \cup C_i| = n - t$, then input $v_i$ to $\rec$ 
        \Else
            \State add $P_j$ to $B_i$
            \If{$|B_i| = t + 1$}
                \State multicast $\bot$
                \State output $\bot$ if you have not output before, and keep running
            \EndIf
        \EndIf
    \EndUpon

    \Upon{receiving $\bot$ from a party $P_j$}
        \State add $P_j$ to $C_i$
        \State if $|C_i| = t + 1$, then output $\bot$ if you have not output before, and keep running
        \State if $|A_i \cup C_i| = n - t$, then input $v_i$ to $\rec$ 
    \EndUpon

    \Upon{outputting $y$ from $\rec$}
        \State input $y$ to $\sra$
    \EndUpon

    \Upon{outputting $y$ from $\sra$}
        \State output $y$ if you have not output before, and keep running
    \EndUpon
\end{algobox}

In $\ca_1$, each party $P_i$ keeps track of three party sets: the set $A_i$ of parties from whom it receives matching hashes (as in $\sra$), as well as the set $B_i$ of parties from whom it receives non-matching hashes and the set $C_i$ of parties from whom it receives the message $\bot$, which a party multicasts when it observes $t + 1$ non-matching hashes and thus learns that it is okay to output $\bot$.

We design $\ca_1$ around the $k$-core predicate defined below (Definition \ref{corepreddef}).

\begin{definition} \label{corepreddef}
    We say that the $k$-core predicate holds in $\ca_1$ if for some $v \in \{0,1\}^\ell$ (a witness of the predicate), there are at least $k$ honest parties who acquire the input $v$ and never multicast $\bot$.
\end{definition}

In particular, we are interested in the $(t+1)$-core predicate due to the Lemmas \ref{predyeslemma} and \ref{prednolemma} below.

\begin{lemma} \label{predyeslemma}
    In $\ca_1$, if the $(t + 1)$-core predicate holds with a witness $v^*$ and no two honest parties' input hashes collide, then each honest party $P_i$ with the $\ca_1$ input $v_i$ lets $v_i$ be its $\rec$ input iff $v_i = v^*$.
\end{lemma}

\begin{lemma} \label{prednolemma}
    In $\ca_1$, if the $(t + 1)$-core predicate does not hold and no two honest parties' input hashes collide, then at least $t + 1$ honest parties multicast $\bot$.
\end{lemma}

The key observation is that in every execution of $\ca_2$, either the $(t + 1)$-core predicate holds with some witness $v^*$, or it does not. In the former case, by Lemma \ref{predyeslemma}, the ($t + 1$ or more) parties with the $\ca_1$ input $v^*$ let $v^*$ be their $\rec$ inputs, while the remaining parties do not input to $\rec$. Consequently, the following happens: Everyone outputs $v^*$ from $\rec$; everyone inputs $v^*$ to $\sra$; everyone outputs $v^*$ from $\sra$; each party outputs $v^*$ unless it has output $\bot$ earlier.

Meanwhile, in the latter case where the $(t + 1)$-core predicate does not hold, Lemma \ref{prednolemma} tells us that $t + 1$ parties multicast $\bot$. These parties' $\bot$ messages let every party $P_i$ obtain $|C_i| = t + 1$ and thus output $\bot$. However, this alone does not imply weak consistency. The parties can input inconsistent bitstrings to $\rec$, and therefore output anything from $\rec$. Still, we get weak consistency thanks to $\sra$. If two parties $P$ and $P'$ respectively output $v \neq \bot$ and $v' \neq \bot$ from $\sra$ and thus output $v$ and $v'$ from $\ca_1$, then $v = v'$ by the consistency of $\sra$. Hence, $\ca_1$ achieves weak consistency.

\begin{proof}[Proof of Lemma \ref{predyeslemma}]

Suppose the $(t + 1)$-core predicate holds with some witness $v^*$. Partition the honest parties into the set $H$ of honest parties with the input $v^*$, and the set $H'$ of honest parties with other inputs. Then, let $H^* \subseteq H$ be the set of the $t + 1$ or more honest parties in $H$ that never multicast $\bot$, and observe that $|H| \geq |H^*| \geq t + 1$ by definition. Now, we have the following: \begin{itemize}
    \item Each $P_i \in H'$ receives a non-matching hash from each $P_j \in H$, and thus adds $P_j$ to $B_i$. So, $P_i$ observes that $|B_i| \geq t + 1$, and thus multicasts $\bot$.
    \item Each $P_i \in H$ receives a matching hash from each $P_j \in H$ and receives $\bot$ from each $P_j \in H'$, which means that $P_i$ adds each party in $H$ to $A_i$ and each party in $H'$ to $C_i$. So, $P_i$ observes that $|A_i \cup C_i| \geq |H \cup H'| \geq n - t$, and thus inputs $v_i = v^*$ to $\rec$.
    \item For each $P_i \in H'$, the parties in $H^*$ send non-matching hashes to $P_i$, which means that $P_i$ does not add them to $A_i$, and the parties in $H^*$ do not multicast $\bot$, which means that $P_i$ does not add them to $C_i$ either. Hence, for $P_i$ it always holds that $|A_i \cup C_i| \leq n - |H^*| < n - t$, which means that $P_i$ does not input $v_i$ to $\rec$. \qedhere
\end{itemize}
    
\end{proof}
    
\begin{proof}[Proof of Lemma \ref{prednolemma}]
Let $v^*$ be the most common honest input (with arbitrary tie-breaking), and let $H$ be the set of honest parties with the input $v^*$. All but at most $t$ of the parties in $H$ multicast $\bot$. Else, $t + 1$ honest parties with the input $v^*$ would never multicast $\bot$, and the $(t + 1)$-core predicate would hold. Moreover, every honest party outside $H$ multicasts $\bot$. To see why, let us choose any $v' \neq v^*$, and see that every honest party with the input $v'$ multicasts $\bot$. Let $H'$ be the set of of honest parties with the input $v'$. There are at least $\ceil[big]{\frac{n - t}{2}}$ honest parties outside $H'$. This is because either \linebreak $|H'| \geq \ceil[big]{\frac{n - t}{2}}$, and so at least $|H| \geq |H'| \geq \ceil[big]{\frac{n - t}{2}}$ honest parties have the input $v^*$, or $|H'| \leq \ceil[big]{\frac{n - t}{2}} - 1$, and so at least $n - t - (\ceil[big]{\frac{n - t}{2}} - 1) \geq \ceil[big]{\frac{n - t}{2}}$ honest parties have inputs other than $v'$. So, assuming no hash collisions, the honest parties in $H'$ receive non-matching hashes from at least $\ceil[big]{\frac{n - t}{2}} \geq t + 1$ honest parties outside $H'$, and thus they multicast $\bot$. Finally, as all honest parties outside $H$ and all but at most $t$ of the parties in $H$ multicast $\bot$, at least $n - 2t \geq t + 1$ honest parties multicast $\bot$. \qedhere
    
\end{proof}

\begin{theorem} \label{watwosec}
    Except with probability less than $2^{-\lambda}$, $\ca_1$ achieves validity, weak consistency and liveness.
\end{theorem}

\begin{proof} We prove $\ca_1$ secure assuming that $\sra$ achieves validity and consistency, and that no two honest parties' input hashes collide in $\ca_1$. These assumptions hold except with probability less than $2^{-\lambda}$ by the union bound since $\sra$ fails with probability less than $2^{-\lambda-1}$ and since the probability that any two honest parties' input hashes collide in $\ca_1$ is (as in $\sra$) less than $2^{-\lambda-1}$.

    \begin{itemize}
        \item \textit{Liveness:} If the $(t + 1)$-core predicate holds with some witness $v^*$, then by Lemma \ref{predyeslemma}, at least $t + 1$ honest parties input $v^*$ to $\rec$, and no honest parties input anything else to $\rec$. By the validity and liveness of $\rec$ and the validity of $\sra$, this is followed by the honest parties all outputting $v^*$ from $\rec$, inputting $v^*$ to $\sra$ and outputting $v^*$ from $\sra$. Hence, every honest party $P_i$ that does not output $\bot$ from $\ca_1$ after obtaining $|B_i| = t + 1$ or $|C_i| = t + 1$ outputs $v^*$ from $\ca_1$ after outputting $v^*$ from $\sra$. On the other hand, if the $(t + 1)$-core predicate does not hold, then by Lemma \ref{prednolemma}, at least $t + 1$ honest parties multicast $\bot$. This ensures that every honest party $P_i$ can obtain $|C_i| \geq t + 1$ and output $\bot$.
        \item \textit{Weak Consistency:} By the consistency of $\sra$, there exists some $y \in \{0,1\}^\ell$ such that every honest $\sra$ output equals $y$. Hence, in $\ca_1$, every honest party $P_i$ either outputs $\bot$ after obtaining $|B_i| = t + 1$ or $|C_i| = t + 1$, or outputs $y$ after outputting $y$ from $\sra$.
        \item \textit{Validity:} If the honest parties have a common input $v^*$, then they send each other matching hashes. So, they do not receive $t + 1$ non-matching hashes, do not multicast $\bot$, and do not receive $\bot$ from $t + 1$ parties. This means that they do not output $\bot$. Moreover, the $(t + 1)$-core predicate holds with the witness $v^*$, which as we argued while proving liveness leads to the honest parties all outputting $v^*$ or $\bot$, or just $v^*$ in this case as they cannot output $\bot$. \qedhere
    \end{itemize}
    
\end{proof}

\paragraph*{Complexity of \textsf{CA}\textsubscript{1}} The round complexity of $\ca_1$ is constant since the round complexities of $\rec$ and $\sra$ are constant. Meanwhile, we get the message complexity $\Theta(n^2)$ and the communication complexity $\Theta(\ell n + n^2\kappa) = \BO(\ell n + n^2(\lambda + \log(\ell n))) = \BO(\ell n + n^2(\lambda + \log n))$ by summing up the the messages sent in $\ca_1$ with the messages sent in $\rec$ and $\sra$. Note that $\BO(\ell n + n^2(\lambda + \log n + \log \ell)) = \BO(\ell n + n^2(\lambda + \log n))$ since $\ell \geq n^2 \implies n^2\log \ell = \BO(\ell n)$ while $\ell < n^2 \implies \log \ell = \BO(\log n)$.

\section{Perfect Security Against \texorpdfstring{$t \leq \frac{n}{3 + \varepsilon}$}{t ≤ n/(3 + ε)} Corruptions} \label{perfsec}

In this section, the parties no longer exchange keyed hashes of their inputs. Instead, for suitable choices of $\delta$, they encode their inputs with $(n, \delta)$-error correcting codes, and exchange indexed symbols of their inputs. These symbols sometimes collide: For any $m, m' \in \{0, 1\}^\ell$ with $\Enc_\delta(m) = (s_1, \dots, s_n)$ and $\Enc_k(m') = (s_1', \dots, s_n')$, there can exist up to $\delta - 1$ indices $j$ such that $s_j = s_j'$, even if $m \neq m'$.\footnote{$\Enc_\delta(m)$ and $\Enc_\delta(m')$ cannot have $\delta$ common symbols, as then these symbols could be decoded into both $m$ and $m'$.} That is why we design $\ca_2$ against $t \leq \frac{n}{3 + \varepsilon}$ faults. When $\varepsilon > 0$ is constant, doing so lets us set $\delta = \Theta(n)$ (for communication efficiency), and yet overcome the resulting collisions. 

Since $\varepsilon$ being arbitrarily large will sometimes be mathematically inconvenient, we will often use $\sigma = \min(1, \varepsilon)$ instead of $\varepsilon$. Note that this gives us $\frac{\varepsilon}{3 + \varepsilon} \geq \frac{\sigma}{4}$, which will be useful in some calculations.

\subsection{Perfectly Secure \texorpdfstring{$\ceil[\big]{\frac{8}{\sigma}}$}{ceil(8/σ)}-Crusader Agreement}

We will use the perfectly secure $\ceil[\big]{\frac{8}{\sigma}}$-crusader agreement protocol $\kca$ below as a preliminary input processing step in $\ca_2$ to bound the number of collisions each honest party has overcome. We obtain $\kca$ by adapting the first two phases of Chen's 4-phase COOL protocol \cite{chen21} to work in an asynchronous network. These phases result in 2-crusader agreement, which Chen then upgrades to crusader agreement. Likewise, we will upgrade $\kca$ to crusader agreement in $\ca_2$.

\begin{algobox}[ who eventually acquires the input $v_i$]{Protocol $\kca$}
    \State let $M_i^0, M_i^1, S_i^0, S_i^1  \gets \{\}, \{\}, \{\}, \{\}$
    \State let $(s_1, \dots, s_n) \gets \Enc_\delta(v_i)$, where $\delta = \ceil[\big]{\frac{\sigma(n - 3t)}{5}}$
    \State to each party $P_j$, send $\msgpair{\SYM}{(s_i, s_j)}$

    \Upon{receiving some $\msgpair{\SYM}{(s_j', s_i')}$ from a party $P_j$ for the first time}
        \State add $P_j$ to $M_i^1$ if $(s_j', s_i') = (s_j, s_i)$, and to $M_i^0$ otherwise
        \State if $|M_i^1| = n - 2t$ and $|M_i^0| < t + 1$, then multicast $\msgpair{\SUC}{1}$
        \State if $|M_i^0| = t + 1$ and $|M_i^1| < n - 2t$, then multicast $\msgpair{\SUC}{0}$
    \EndUpon

    \Upon{receiving some $\msgpair{\SUC}{b}$ from a party $P_j$ for the first time}
        \State add $P_j$ to $S_i^b$
    \EndUpon

    \When{$|M_i^1 \cap S_i^1| \geq n - 2t$}
        \State output $v_i$ if you have not output already
    \EndWhen

    \When{$|M_i^0 \cup S_i^0| \geq t + 1$}
        \State output $\bot$ if you have not output already
    \EndWhen
\end{algobox}

\begin{lemma} \label{wclemma}
    In $\kca$, if an honest $P_i$ outputs $v_i$, then at least $\frac{\sigma n}{8}$ honest parties have the input $v_i$.
\end{lemma}

The $\ceil[\big]{\frac{8}{\sigma}}$-weak consistency of $\kca$ will follow from Lemma \ref{wclemma}, which holds because if an honest party $P_i$ outputs its input $v_i$, then for at least $\frac{\sigma n}{8} - 1$ other honest parties $P_j$ it holds that $\Enc_{\delta}(v_i)$ and $\Enc_{\delta}(v_j)$ have at least $\delta$ matching symbols, and thus that $v_i = v_j$. This fact does not follow from a typical quorum intersection argument. To prove it, we need the following lemma:

\begin{restatable}{lemma}{graphlemma} \label{graphlemma}
    Let $G = (V, E)$ be a graph with loops but no multiedges, with $|V| = n - t$ vertices where $n > 3t \geq 0$. For any $a \in (0, 1)$ and any set $C \subseteq V$ of $n - 3t$ vertices such that every $v \in C$ has a closed neighborhood (neighborhood including $v$) of size at least $n - 3t$, the set $D \subseteq V$ of vertices $v$ with at least $a(n - 3t)$ neighbors in $C$ (including $v$ if $v \in C$) is of size at least $n - (3 + \frac{2a}{1 - a})t$.
\end{restatable}

Since Lemma \ref{graphlemma} is a variant of ``Lemma 8'' in \cite{chen2020} (the full arXiv version of Chen's COOL protocol) with adjusted parameters and a slightly different statement, we provide its proof in the \hyperref[kwaappendix]{appendix}. 

Below, for any honest parties $P_i$ and $P_j$ where $\Enc_\delta(v_i) = (s_1, \dots, s_n)$ and $\Enc_\delta(v_i) = (s_1, \dots, s_n)$, we say that  $P_i$ and $P_j$ are matching if $(s_j', s_i') = (s_j, s_i)$, and non-matching otherwise.

\begin{proof}[Proof of Lemma \ref{wclemma}]
    Suppose some honest party $P_i$ outputs its input $v_i$ from $\kca$. Let $H$ be any set of $n - t$ honest parties that includes $P_i$, and form a graph $G$ with the parties in $H$ as vertices, with an edge between every $P_j, P_k \in H$ that are matching, including a loop edge at every $P_j \in H$.

    Since $P_i$ outputs $v_i$, we have $|M_i^1 \cap S_i^1| \geq n - 2t$, and thus $|M_i^1 \cap S_i^1 \cap H| \geq n - 3t$. Therefore, we can let $C$ be any subset of $M_i^1 \cap S_i^1 \cap H$ of size $n - 3t$. Every $P_j \in C$ is an honest party matching $P_i$ that has multicast $\msgpair{\SUC}{1}$, which means that $|M_j^1| \geq n - 2t$ and thus that $|M_j^1 \cap H| \geq n - 3t$. In other words, every $P_j$ in $C$ has at least $|M_j^1 \cap H| \geq n - 3t$ neighbors in $G$.

    Let $D$ be the set of parties in $H$ which in $G$ have at least $\frac{\sigma}{5}(n - 3t)$ neighbors in $C$. By Lemma \ref{graphlemma}, we have $|D| \geq n - (3 + \frac{2\sigma/5}{1 - \sigma/5})t$, which is at least $n - (3 + \frac{2\sigma/5}{4/5})t = n - (3 + \frac{\sigma}{2})t$ since $\sigma \leq 1$. Now, since $n = bt$ for some $b \geq 3 + \sigma$, we have $\frac{|D|}{n} \geq \frac{bt - (3 + \sigma / 2)t}{bt} = \frac{b - (3 + \sigma / 2)}{b}$. This fraction is increasing for $b \in [3 + \sigma, \infty)$, which means that $\frac{|D|}{n} \geq \frac{b - (3 + \sigma / 2)}{b} \geq \frac{3 + \sigma - (3 + \sigma/2)}{3 + \sigma} = \frac{\sigma/2}{3 + \sigma} \geq \frac{\sigma / 2}{4} = \frac{\sigma}{8}$. So, $|D| \geq \frac{\sigma n}{8}$.

    Finally, for each $P_j \in D \setminus \{P_i\}$, let $C_j$ be the set of $P_j$'s neighbors in $C$, including $P_j$ if $P_j \in C$. Each $P_k \in C_j$ with the input $v_k$ is matching with both $P_i$ and $P_j$, which (even if $P_k \in \{P_i, P_j\}$) means that the $k^{\text{th}}$ symbol of $\Enc_\delta(v_k)$ is equal to the $k^{\text{th}}$ symbols of both $\Enc_\delta(v_i)$ and $\Enc_\delta(v_j)$. Since $|C_j| \geq \frac{\sigma}{5}(n - 3t)$ implies that $|C_j| \geq \ceil[\big]{\frac{\sigma(n - 3t)}{5}} = \delta$, we conclude that there are at least $\delta$ indices $k \in [n]$ such that $\Enc_\delta(v_i)$ and $\Enc_\delta(v_j)$ have matching $k^{\text{th}}$ symbols, and thus that $v_i = v_j$. So, there are at least $|(D \setminus \{P_i\}) \cup \{P_i\}| \geq |D| \geq \frac{\sigma n}{8}$ honest parties with the input $v_i$.
\end{proof}

\begin{theorem}
    The protocol $\kca$ achieves liveness, validity and $\ceil[\big]{\frac{8}{\sigma}}$-weak consistency.
\end{theorem}

\begin{proof}
    In $\kca$, each honest party $P_i$ tracks the set $M_i^1$ of parties from which it receives matching symbol pairs, the set $M_i^0$ of parties from which it receives non-matching symbol pairs, the set $S_i^1$ of parties from which it receives the positive success indicator $\msgpair{\SUC}{1}$, and the set $S_i^0$ of parties from which it receives the negative success indicator $\msgpair{\SUC}{0}$.

    An honest party $P_i$ receives a symbol pair from every honest $P_j$, and thus adds every honest $P_j$ to either $M_i^1$ or $M_i^0$. Therefore, eventually $|M_i^1 \cup M_i^0| \geq n - t$ holds, and hence either $|M_i^1| \geq n - 2t$ holds or $|M_i^0| \geq t + 1$ holds. If $|M_i^1| \geq n - 2t$ holds first, then $P_i$ multicasts $\msgpair{\SUC}{1}$. Otherwise, if $|M_i^0| \geq t + 1$ holds first, then $P_i$ multicasts $\msgpair{\SUC}{0}$.

    Since the honest parties all multicast success indicators, an honest party $P_i$ adds every honest party $P_j$ to not just one of $M_i^1$ and $M_i^0$ but also to one of $S_i^1$ or $S_i^0$ after receiving $P_j$'s success indicator, which means that eventually $|(M_i^1 \cap S_i^1) \cup (M_i^0 \cup S_i^0)| \geq n - t$. So, eventually either $|M_i^1 \cap S_i^1| \geq n - 2t$ holds and $P_i$ outputs $v_i$, or $|M_i^0 \cup S_i^0| \geq t + 1$ holds and $P_i$ outputs $\bot$. Either way, we have liveness. 
    
    As for validity, suppose the honest parties have a common input $v^*$. Then, they all match each other, which gives us the chain of implications that they do not add each other to their $M^0$ sets, do not obtain $M^0$ sets of size at least $t + 1$, do not multicast $\msgpair{\SUC}{0}$, and do not add each other to their $S^0$ sets either. Finally, as an honest party $P_i$ never places honest parties in $M_i^0 \cup S_i^0$, it never obtains $|M_i^0 \cup S_i^0| \geq t + 1$ and thus never outputs $\bot$, which means that it must output its own input $v^*$. \qedhere

    Lastly, $\ceil[\big]{\frac{8}{\sigma}}$-weak consistency follows easily from Lemma \ref{wclemma}. By the lemma, every honest non-$\bot$ output $y$ implies the existence of at least $\frac{\sigma n}{8}$ honest parties with the input $y$. Since there are at most $n$ honest parties, there is no room for more than $\frac{n}{(\sigma n) / 8} = \frac{8}{\sigma}$ distinct non-$\bot$ outputs.
\end{proof}

\paragraph*{Complexity of \textsf{KCA}} The protocol $\kca$ consists of 2 rounds: one for symbol exchanging, and one for success indicator exchanging. Its message complexity is $\Theta(n^2)$. As for the communication complexity, since we use $(n, \delta)$-error correction with $\delta = \ceil[\big]{\frac{\sigma(n - 3t)}{5}} \geq \frac{\sigma n}{5} \cdot \frac{n - 3t}{n} \geq \frac{\sigma n}{5} \cdot \frac{\varepsilon}{3 + \varepsilon} \geq \frac{\sigma n}{5} \cdot \frac{\sigma}{4} = \frac{\sigma^2 n}{20}$, we have $\frac{n}{\delta} = \BO(\sigma^{-2})$, which causes each symbol to be of size $\BO\bigl(\frac{\ell}{n \cdot \min(1, \varepsilon^2)} + \max\bigl(1, \log \frac{1}{\varepsilon}\bigr) \bigr)$ and thus causes the communication complexity to be $\BO\bigl(\frac{\ell n}{\min(1, \varepsilon^2)} + n^2\max\bigl(1, \log \frac{1}{\varepsilon}\bigr) \bigr)$.

\paragraph*{Discussion.} The protocol $\kca$ resembles the first two synchronous rounds of COOL \cite{chen21}, though in COOL after these rounds a party $P_i$ outputs $v_i$ if $|M_i^1 \cap S_i^1|\ \geq n - t$, and $\bot$ otherwise. For liveness in asynchrony, we reduce the $n - t$ threshold to $n - 2t$, and compensate by requiring $t \leq \frac{n}{3 + \varepsilon}$ to still achieve $k_\varepsilon$-weak consistency for some $k_\varepsilon \geq 1$ that depends only on $\varepsilon$. Comparable but different adaptations for asynchrony can be found in A-COOL \cite{lichen21}, which sticks closer to COOL's design and therefore only tolerates $t < \frac{n}{5}$ faults. Note that COOL-based reliable broadcast protocols \cite{long22, chen24} can keep the threshold $n - t$ and thus the resilience $t < \frac{n}{3}$ despite asynchrony since reliable broadcast requires the parties to output only if the sender is honest, which intuitively corresponds to the case where the parties run $\kca$ with \mbox{a common input and thus all add each other to their $M^1$ and $S^1$ sets.}

\subsection{Perfectly Secure Reliable Agreement}

The assumption $t \leq \frac{n}{3 + \varepsilon}$ leads to the simple perfectly secure reliable agreement protocol $\pra$ below where each party $P_i$ with the input $v_i$ computes $(s_1, \dots, s_n) \gets \Enc_{n - 3t}(v_i)$, multicasts $s_i$, and outputs $v_i$ after from $n - t$ parties $P_j$ it receives the matching symbol $s_j$. Note that each symbol exchanged in $\pra$ is of size $\BO\bigl(\frac{\ell}{n \cdot \min(1, \varepsilon)} + \max\bigl(1, \log \frac{1}{\varepsilon}\bigr) \bigr)$ as $t \leq \frac{n}{3 + \varepsilon}$ implies $n - 3t \geq \frac{\varepsilon n}{3 + \varepsilon} = \Omega(\min(1, \varepsilon)n)$.

\begin{algobox}[ who eventually acquires the input $v_i$]{Protocol $\pra$}
    \State let $A_i \gets \{P_i\}$
    \State let $(s_1, \dots, s_n) \gets \Enc_{n - 3t}(v_i)$
    \State multicast $\msgpair{\SYM}{s_i}$

    \Upon{receiving some $\msgpair{\SYM}{s_j'}$ from any other party $P_j$ for the first time}
        \State if $s_j' = s_j$, then add $P_j$ to $A_i$
        \State if $|A_i| = n - t$, then output $v_i$, and keep running
    \EndUpon
\end{algobox}

\begin{theorem} \label{prasec}
    The protocol $\pra$ achieves validity and consistency.
\end{theorem}

\begin{proof} If the honest parties have a common input $v^*$, then they all send each other matching symbols. So, every honest $P_i$ adds every honest $P_j$ to $A_i$, obtains $|A_i| = n - t$, and outputs $v_i = v^*$. So, we have validity. Meanwhile, for consistency, observe that if two honest parties $P_i$ and $P_j$ output $v_i$ and $v_j$, then $|A_i| = n - t$ and $|A_j| = n - t$, which implies $|A_i \cap A_j| \geq n - 2t$. So, there are $n - 3t$ honest parties $P_k \in A_i \cap A_j$ who sent both $P_i$ and $P_j$ matching symbols, i.e.\ $n - 3t$ honest parties $P_k$ such that $k^{\text{th}}$ symbol of $\Enc_{n - 3t}(v_k)$ matches the $k^{\text{th}}$ symbols of both $\Enc_{n - 3t}(v_i)$ and $\Enc_{n - 3t}(v_j)$. This means that $\Enc_{n - 3t}(v_i)$ and $\Enc_{n - 3t}(v_j)$ have $n - 3t$ matching symbols, and thus that $v_i = v_j$. \qedhere

\end{proof}

\paragraph*{Complexity of \textsf{PRA}} The protocol $\pra$ consists of 1 symbol exchange round, $n^2$ symbol messages, and $\BO\bigl(\frac{\ell n}{\min(1, \varepsilon)} + n^2\max\bigl(1, \log \frac{1}{\varepsilon}\bigr) \bigr)$ bits of communication.

\subsection{Perfectly Secure Crusader Agreement}

We are finally ready for the perfectly secure crusader agreement protocol $\ca_2$.

\begin{algobox}[ who eventually acquires the input $v_i$]{Protocol $\ca_2$}
    \State join an instance of $\kca$, an instance of $\rec$ and an instance of $\pra$
    \State input $v_i$ to $\kca$, and let $z_i$ be the output you obtain from it
    \State if $z_i = \bot$, then multicast $\msgpair{\SYM}{\bot}$, multicast $\bot$ and output $\bot$
    \State \textit{\textcolor{red}{Run the code below iff $z_i \neq \bot$}}
    \State let $A_i, B_i, C_i \gets \{P_i\}, \{\}, \{\}$
    \State let $(s_1, \dots, s_n) \gets \Enc_\delta(z_i)$, where $\delta = \ceil[\big]{\frac{\sigma(n - 3t)}{16}}$
    \State multicast $\msgpair{\SYM}{s_i}$

    \Upon{receiving some $\msgpair{\SYM}{s_j'}$ from any other party $P_j$ for the first time}
        \If{$s_j' \neq s_j$}
            \State add $P_j$ to $A_i$
            \State if $|A_i \cup C_i| = n - t$, then input $v_i$ to $\rec$ 
        \Else
            \State add $P_j$ to $B_i$
            \If{$|B_i| = t + 1$}
                \State multicast $\bot$
                \State output $\bot$ if you have not output before, and keep running
            \EndIf
        \EndIf
    \EndUpon

    \Upon{receiving $\bot$ from a party $P_j$}
        \State add $P_j$ to $C_i$
        \State if $|C_i| = t + 1$, then output $\bot$ if you have not output before, and keep running
        \State if $|A_i \cup C_i| = n - t$, then input $v_i$ to $\rec$ 
    \EndUpon

    \Upon{outputting $y$ from $\rec$}
        \State input $y$ to $\pra$
    \EndUpon

    \Upon{outputting $y$ from $\pra$}
        \State output $y$ if you have not output before, and keep running
    \EndUpon
\end{algobox}

The protocol $\ca_1$ uses the same design as $\ca_1$, with two changes: \begin{itemize}
    \item We replace the key and hash exchange of $\ca_1$ with the exchange of symbols, and accordingly use the symbol-based reliable agreement protocol $\pra$ instead of the hash-based $\sra$.
    \item The protocol $\ca_2$ begins with a new input processing step where the parties run $\kca$ first, and then run the rest of the protocol with their $\kca$ outputs instead of their raw inputs. The fact that $\kca$ only permits at most $\ceil[\big]{\frac{8}{\sigma}}$ distinct honest non-$\bot$ outputs lets us bound the number of symbol collisions in $\ca_2$. If an honest party outputs $\bot$ from $\kca$, then it outputs $\bot$ from $\ca_2$, and tells everyone to also output $\bot$ by multicasting $\msgpair{\SYM}{\bot}$ and $\bot$.
\end{itemize}

\begin{theorem} \label{wathreesec}
    The protocol $\ca_2$ achieves validity, weak consistency and liveness.
\end{theorem}

To prove Theorem \ref{wathreesec}, we bound the number of colliding honest parties with Lemma \ref{limlemma} below. In $\ca_2$, we use an $(n, \delta)$-error correcting code where $\delta = \ceil[\big]{\frac{\sigma(n - 3t)}{16}}$ to ensure that this lemma holds.

\begin{lemma} \label{limlemma}
    In $\ca_2$, if an honest $P_i$ outputs some $z_i \neq \bot$ from $\kca$, then there are at most $\floor[\big]{\frac{n - 3t - 1}{2}}$ ``colliding'' honest $P_j$ such that $P_j$ outputs some $z_j \neq z_i$ from $\kca$ but sends a matching symbol to $P_i$. 
\end{lemma}

\begin{proof}[Proof of Lemma \ref{limlemma}]
    If an honest party $P_i$ outputs some $z_i \neq \bot$ from $\kca$, then each honest $P_j$ \linebreak colliding with $P_i$ implies a unique output-index pair $(z_j, j)$ such that $P_j$'s $\kca$ output $z_j$ is one of the at most $\ceil[\big]{\frac{8}{\sigma}} - 1$ honest $\kca$ outputs besides $\bot$ and $z_i$, and $j$ is one of the at most $\delta - 1$ indices $k$ such that $\Enc_{\delta}(z_i)$ and $\Enc_{\delta}(z_j)$ have the same $k^{\text{th}}$ symbol. So, there are only $\big(\ceil[\big]{\frac{8}{\sigma}} - 1\big)(\delta - 1) < \frac{8}{\sigma} \cdot \frac{\sigma(n - 3t)}{16} = \frac{n - 3t}{2}$ ways, or in other words at most $\floor[\big]{\frac{n - 3t - 1}{2}}$ ways, to choose a colliding party.
\end{proof}

Below, we prove Theorem \ref{wathreesec} by generalizing our analysis of $\ca_1$ to account for the colliding honest parties in $\ca_2$. Like the security of $\ca_1$, the security of $\ca_2$ revolves around the $k$-core predicate. Since $\ca_2$ treats $\kca$ outputs as inputs for the rest of the protocol, the $k$-core predicate for $\ca_2$ below (Definition \ref{predperfect}) concerns $\kca$ outputs rather than $\ca_2$ inputs.

\begin{definition} \label{predperfect}
    We say that the $k$-core predicate holds in $\ca_2$ if for some $z \in \{0,1\}^\ell$ (a witness of the predicate), there are at least $k$ honest parties who output $z$ from $\kca$ and never multicast $\bot$.
\end{definition}

While we used the $(t + 1)$-core predicate for $\ca_1$, we will use the $\ceil[\big]{\frac{n - t}{2}}$-core predicate for $\ca_2$. The gap between $\ceil[\big]{\frac{n - t}{2}}$ and $t + 1$ will be how $\ca_2$ tolerates collisions. Note that $\ceil[\big]{\frac{n - t}{2}} - (t + 1) = \floor[\big]{\frac{n - 3t - 1}{2}}$, which is non-coincidentally the number of collisions each honest party must overcome in $\ca_2$. The key lemmas for $\ca_2$ are the Lemmas \ref{predyescol} and \ref{prednocol} below, which are respectively the analogues of the Lemmas \ref{predyeslemma} and \ref{prednolemma} with which we proved $\ca_1$ secure.

\begin{lemma} \label{predyescol}
    In $\ca_2$, if the $\ceil[\big]{\frac{n - t}{2}}$-core predicate holds with a witness $z^*$, then each honest party $P_i$ with the $\kca$ output $z_i$ lets $z_i$ be its $\rec$ input iff $z_i = z^*$.
\end{lemma}

\begin{lemma} \label{prednocol}
    In $\ca_2$, if the $\ceil[\big]{\frac{n - t}{2}}$-core predicate does not hold, then at least $t + 1$ honest parties multicast $\bot$.
\end{lemma}

In our proofs below, for each honest party $P_i$ we let $E_i$ denote the set of honest parties that collide with $P_i$ (i.e.\ obtain $\kca$ outputs other than $P_i$'s $\kca$ output $z_i$ but send $P_i$ matching symbols), and use the fact that $|E_i| \leq \floor[\big]{\frac{n - 3t - 1}{2}}$ for all honest $P_i$.

\begin{proof}[Proof of Lemma \ref{predyescol}]

Suppose the $\ceil[\big]{\frac{n - t}{2}}$-core predicate holds with a witness $z^*$. Partition the honest parties into the set $H$ of honest parties that output $z^*$ from $\kca$, and the set $H'$ of honest parties with other $\kca$ outputs. Then, let $H^* \subseteq H$ be the set of $\ceil[\big]{\frac{n - t}{2}}$ or more honest parties in $H$ that never multicast $\bot$, and observe that $|H| \geq |H^*| \geq \ceil[\big]{\frac{n - t}{2}}$ by definition. Now, since each honest party $P_i$ observes at most $|E_i| \leq \floor[\big]{\frac{n - 3t - 1}{2}}$ collisions, we have the following: \begin{itemize}
    \item Each $P_i \in H'$ receives a non-matching symbol from each $P_j \in H \setminus E_i$, and thus adds $P_j$ to $B_i$. So, $P_i$ observes that $|B_i| \geq |H| - |E_i| \geq \ceil[\big]{\frac{n - t}{2}} - \floor[\big]{\frac{n - 3t - 1}{2}} \geq t + 1$, and thus multicasts $\bot$.
    \item Each $P_i \in H$ receives a matching symbol from each $P_j \in H$ and receives $\bot$ from each $P_j \in H'$, which means that $P_i$ adds each party in $H$ to $A_i$ and each party in $H'$ to $C_i$. So, $P_i$ observes that $|A_i \cup C_i| \geq |H \cup H'| \geq n - t$, and thus inputs $z_i = z^*$ to $\rec$.
    \item For each $P_i \in H'$, the parties in $H^* \setminus E_i$ send non-matching symbols to $P_i$, which means that $P_i$ does not add them to $A_i$, and the parties in $H^* \setminus E_i$ do not multicast $\bot$, which means that $P_i$ does not add them to $C_i$ either. Hence, for $P_i$ it always holds that $|A_i \cup C_i| \leq n - |H^* \setminus E_i| \leq n - \big(\ceil[\big]{\frac{n - t}{2}} - \floor[\big]{\frac{n - 3t - 1}{2}}\big) < n - t$, which means that $P_i$ does not input $v_i$ to $\rec$. \qedhere
\end{itemize}
\end{proof}

\begin{proof}[Proof of Lemma \ref{prednocol}]
Let $z^*$ be the most common honest $\kca$ output (with ties broken arbitrarily), and let $H$ be the set of honest parties with the $\kca$ output $z^*$. All but at most $\ceil[big]{\frac{n - t}{2}} - 1$ of the parties in $H$ multicast $\bot$. Else, $\ceil[big]{\frac{n - t}{2}}$ honest parties with the $\kca$ output $z^*$ would never multicast $\bot$, and the $\ceil[big]{\frac{n - t}{2}}$-core predicate would hold. Moreover, every honest party outside $H$ multicasts $\bot$. To see why, let us choose any $z' \neq z^*$, and see that every honest party with the $\kca$ output $z'$ multicasts $\bot$. Let $H'$ be the set of of honest parties with the $\kca$ output $z'$. There are at least $\ceil[big]{\frac{n - t}{2}}$ honest parties outside $H'$. This is because either $|H'| \geq \ceil[big]{\frac{n - t}{2}}$, and therefore at least $|H| \geq |H'| \geq \ceil[big]{\frac{n - t}{2}}$ honest parties output $z^*$ from $\kca$, or $|H'| \leq \ceil[big]{\frac{n - t}{2}} - 1$, and therefore at least $n - t - (\ceil[big]{\frac{n - t}{2}} - 1) \geq \ceil[big]{\frac{n - t}{2}}$ honest parties obtain outputs other than $z'$ from $\kca$. Consequently, by the $\floor[big]{\frac{n - 3t - 1}{2}}$ bound on symbols collisions, the honest parties in $H'$ receive non-matching symbols from at least $\ceil[big]{\frac{n - t}{2}} - \floor[big]{\frac{n - 3t - 1}{2}} \geq t + 1$ honest parties outside $H'$, and thus they multicast $\bot$. Finally, as all honest parties outside $H$ and all but at most $\ceil[\big]{\frac{n - t}{2}} - 1$ of the parties in $H$ multicast $\bot$, at least $n - t - (\ceil[\big]{\frac{n - t}{2}} - 1) \geq \ceil[\big]{\frac{n - t}{2}} \geq t + 1$ honest parties multicast $\bot$. \qedhere

\end{proof}

\begin{proof}[Proof of Theorem \ref{wathreesec}] \hfill

    \begin{itemize}
        \item \textit{Liveness:} If the $\ceil[\big]{\frac{n - t}{2}}$-core predicate holds with some witness $z^*$, then by Lemma \ref{predyescol}, at least $\ceil[\big]{\frac{n - t}{2}} \geq t + 1$ honest parties input $z^*$ to $\rec$, and no honest parties input anything else to $\rec$. By the validity and liveness of $\rec$ and the validity of $\pra$, this is followed by the honest parties all outputting $z^*$ from $\rec$, inputting $z^*$ to $\pra$ and outputting $z^*$ from $\pra$. Hence, every honest party $P_i$ that does not output $\bot$ from $\ca_2$ after outputting $\bot$ from $\kca$ or after obtaining $|B_i| = t + 1$ or $|C_i| = t + 1$ outputs $v^*$ from $\ca_2$ after outputting $v^*$ from $\pra$. On the other hand, if the $\ceil[\big]{\frac{n - t}{2}}$-core predicate does not hold, then by Lemma \ref{prednocol}, at least $t + 1$ honest parties multicast $\bot$. This ensures that every honest $P_i$ can obtain $|C_i| \geq t + 1$ and thus output $\bot$.
        \item \textit{Weak Consistency:} By the consistency of $\pra$, there exists some $y \in \{0,1\}^\ell$ such that every honest $\pra$ output equals $y$. Therefore, in $\ca_2$, every honest party $P_i$ either outputs $y$ after outputting $y$ from $\pra$, or outputs $\bot$ from $\ca_2$ after either outputting $\bot$ from $\kca$ or obtaining $|B_i| = t + 1$ or $|C_i| = t + 1$.
        \item \textit{Validity:} If the honest parties have a common input $v^*$, then they output $v^*$ from $\kca$, and then send each other matching symbols. So, they do not receive $t + 1$ non-matching symbols, do not multicast $\bot$, and do not receive $\bot$ from $t + 1$ parties. This means that they do not output $\bot$. Moreover, the $\ceil[\big]{\frac{n - t}{2}}$-core predicate holds with the witness $v^*$, which as we argued while proving liveness leads to the honest parties all outputting $v^*$ or $\bot$, or just $v^*$ in this case as they cannot output $\bot$. \qedhere
    \end{itemize}
    
\end{proof}

\paragraph*{Complexity of \textsf{CA}\textsubscript{2}} The round complexity of $\ca_2$ is constant as $\kca$, $\rec$ and $\pra$ are constant-round protocols. Meanwhile, we get the message complexity $\Theta(n^2)$ and the communication complexity $\BO\bigl(\frac{\ell n}{\min(1, \varepsilon^2)} + n^2\max\bigl(1, \log \frac{1}{\varepsilon}\bigr) \bigr)$ by csumming up the the messages sent in $\ca_2$ with the messages sent in $\kca$, $\rec$ and $\pra$. Note that the symbols the parties send in $\ca_2$ outside $\kca$, $\rec$ and $\pra$ are of size $\BO\bigl(\frac{\ell}{n \cdot \min(1, \varepsilon^2)} + \max\bigl(1, \log \frac{1}{\varepsilon}\bigr) \bigr)$, as $\delta = \ceil[\big]{\frac{\sigma(n - 3t)}{16}} \geq \frac{\sigma n}{16} \cdot \frac{n - 3t}{n} \geq \frac{\sigma n}{16} \cdot \frac{\varepsilon}{3 + \varepsilon} \geq \frac{\sigma n}{16} \cdot \frac{\sigma}{4} = \frac{\sigma^2 n}{64}$ in $\ca_2$ implies $\frac{n}{\delta} = \BO(\sigma^{-2})$.

\bibliographystyle{ACM-Reference-Format}
\bibliography{refs}

\appendix

\section{The Reconstruction Protocol \texorpdfstring{$\rec$}{REC}} \label{recsec}

The protocol $\rec$ is based on the ADD protocol in \cite{dxr21}. We have removed the $\bot$ input for parties without bitstring inputs, as we use $\rec$ in contexts where some parties never learn if they will acquire bitstring inputs or not. Besides this, $\rec$ would be equivalent to the ADD protocol if we removed the lines \ref{line:l17}-\ref{line:l20}, and modified line \ref{line:l16} to make a party $P_i$ output $y_i$ (but not terminate) directly after setting $y_i \gets y$. In fact, such a shorter version of $\rec$ would suffice for $\ca_1$ and $\ca_2$. The lines \ref{line:l17}-\ref{line:l18} where we make $P_i$ send $\MINE$ and $\YOURS$ symbols based on its output $y_i$ and the lines \ref{line:l19}-\ref{line:l20} where we make $P_i$ wait until it has received $2t + 1$ $\YOURS$ messages to output $y_i$ and terminate are our changes to ensure that $\rec$ terminates with totality. Note that one can find similar \linebreak changes in reliable broadcast protocols with steps based on the same ADD protocol \cite{long22, chen24}.

\begin{algobox}[ who might acquire an input $v_i$]{Protocol $\rec$}
    \State let $(z_1, \dots, z_n) \gets (\bot, \dots, \bot)$\label{line:l1}
    \State let $y_i \gets \bot$\label{line:l2}

    \Upon{acquiring the input $v_i$}\label{line:l3}
        \State $(s_1, \dots, s_n) \gets \Enc_{n - 2t}(v^*)$\label{line:l4}
        \State multicast $\msgpair{\MINE}{s_i}$ unless you multicast a $\MINE$ message before\label{line:l5}
        \State send $\msgpair{\YOURS}{s_j}$ to each party $P_j$ unless you sent $\YOURS$ messages before\label{line:l6}
    \EndUpon

    \Upon{for some $s$ receiving the message $\msgpair{\YOURS}{s}$ from $t + 1$ parties}\label{line:l7}
        \State multicast $\msgpair{\MINE}{s}$ unless you multicast a $\MINE$ message before\label{line:l8}
    \EndUpon

    \Upon{receiving some $\msgpair{\MINE}{s}$ from a party $P_j$ for the first time, if $y_i \neq \bot$}\label{line:l9}
        \State $z_j \gets s$\label{line:l10}
        \If{at least $n - t$ of $z_1, \dots, z_n$ are not $\bot$}\label{line:l11}
            \State $y \gets \Dec_{n - 2t}(z_1, \dots, z_n)$\label{line:l12}
            \If{$y \neq \bot$}\label{line:l13}
                \State $(s_1, \dots, s_n) \gets \Enc_{n - 2t}(y)$\label{line:l14}
                \If{$z_k = s_k$ for at least $n - t$ indices $k \in [n]$}\label{line:l15}
                    \State $y_i \gets y$ \Comment{ $y$ is the correct online error correction output}\label{line:l16}
                    \State multicast $\msgpair{\MINE}{s_i}$ unless you multicast a $\MINE$ message before\label{line:l17}
                    \State send $\msgpair{\YOURS}{s_j}$ to each party $P_j$ unless you sent $\YOURS$ messages before\label{line:l18}
                \EndIf
            \EndIf
        \EndIf
    \EndUpon

    \When{$y_i \neq \bot$ and you have received $\YOURS$ messages from $2t + 1$ parties}\label{line:l19}
        \State output $y_i$ and terminate\label{line:l20}
    \EndWhen
\end{algobox}

On lines \ref{line:l9}-\ref{line:l16} we implement online error correction \cite{bcg93}, a well-known technique for message reconstruction in asynchrony where a party $P_i$ collects symbols from the other parties one-by-one. Whenever $P_i$ has collected $n - t + k$ symbols for any $k \in \{0, \dots, t\}$, it runs $\Dec_{n - 2t}$ on the collected symbols to decode them into a potential output $y$, repeatedly with each new collected symbol until when $y$ is verifiably ``correct,'' which is when $y \neq \bot$ and at least $n - t$ of the collected symbols are correct w.r.t\ $y$. Crucially, if with respect to some $v^* \in \{0, 1\}^\ell$ the symbols $P_i$ collects from the honest parties are all correct (i.e.\ the symbol $z_j$ from any honest $P_j$ is the $j^{\text{th}}$ symbol of $\Enc_{n - 2t}(v^*)$), then the final output $y$ is $v^*$ since $P_i$ verifies that it received at least $n - 2t$ honest symbols that are correct w.r.t.\ $y$, and thus verifies that $\Enc_{n - 2t}(y)$ and $\Enc_{n - 2t}(v^*)$ have at least $n - 2t$ matching symbols.

\begin{lemma} \label{reclemma}
    In $\rec$, if some $v^*$ where $\Enc_{n - 2t}(v^*) = (s_1^*, \dots, s_n^*)$ is the only input the honest parties can acquire, then the honest parties can only act validly w.r.t.\ $v^*$. That is, if an honest $P_i$ outputs $y$, then $y = v^*$, if it multicasts $\msgpair{\MINE}{s}$, then $s = s_i^*$, and if it sends $\msgpair{\YOURS}{s}$ to any party $P_j$, then $s =  s_j^*$.
\end{lemma}

\begin{proof}
    Suppose this lemma fails in some execution of $\rec$. Consider the first moment in this execution where some honest $P_i$ violates the lemma by acting invalidly w.r.t.\ $v^*$. Below, we consider the three cases that can lead to $P_i$ violating the lemma, and show contradictorily that none of them work. We use the fact that since $P_i$'s lemma violation is the first one in the execution, $P_i$ violates the lemma when no honest $P_j$ has sent a $\MINE$ message other than $\msgpair{\MINE}{s_j^*}$ or a $\YOURS$ message other than $\msgpair{\YOURS}{s_i^*}$ to $P_i$. \begin{itemize}
        \item $P_i$ cannot violate the lemma on lines \ref{line:l4}-\ref{line:l5} by sending invalid $\MINE$ or $\YOURS$ symbols after acquiring an input other than $v^*$, because $v^*$ is the only input an honest party can acquire.
        \item $P_i$ cannot violate the lemma on line \ref{line:l7} by invalidly multicasting $\msgpair{\MINE}{s}$ where $s \neq s_i^*$ after receiving $\msgpair{\YOURS}{s}$ from $t + 1$ parties. If $P_i$ has received $\msgpair{\YOURS}{s}$ from $t + 1$ parties, then it has received $\msgpair{\YOURS}{s}$ from an honest party, which means that $s = s_i^*$.
        \item $P_i$ cannot violate the lemma on lines \ref{line:l16}-\ref{line:l18} after obtaining an incorrect online error correction output $y \neq v^*$. This is because $P_i$ verifies on line \ref{line:l13} after computing $y = \Dec_{n - 2t}(z_1, \dots, z_n)$ that $y \neq \bot$, and verifies on line \ref{line:l15} after computing $(s_1, \dots, s_n) = \Enc_{n - 2t}(y)$ that there are at least $n - t$ indices $k \in [n]$ such that $z_k = s_k$. If these verifications succeed, then for at least $n - 2t$ indices $k \in [n]$ such that $P_k$ is an honest party, it holds that $z_k = s_k = s_k^*$, where $s_k = s_k^*$ since $s_k$ is the $\MINE$ symbol from the honest party $P_k$. This implies that $\Enc_{n - 2t}(y)$ and $\Enc_{n - 2t}(v^*)$ have at least $n - 2t$ symbols in common, and thus that $y = v^*$. \qedhere
    \end{itemize}
\end{proof}

\begin{theorem}
    The protocol $\rec$ achieves validity, liveness (with termination) and totality.
\end{theorem}

\begin{proof}
    Suppose some $v^*$ where $\Enc_{n - 2t}(v^*) = (s_1^*, \dots, s_n^*)$ is the only input the honest parties can acquire, as the properties of $\rec$ all require such a $v^*$ to guarantee anything. Lemma \ref{reclemma} implies validity when combined with the observation that before some honest party acquires the input $v^*$, the honest parties do not send each other $\MINE$ or $\YOURS$ messages, and therefore they do not receive enough messages to terminate. So, let us prove liveness (with termination) and totality below, using the fact that the honest parties act validly w.r.t.\ $v^*$.

    Suppose $t + 1$ honest parties send everyone $\YOURS$ messages. These messages ensure that every honest $P_i$ multicasts $\msgpair{\MINE}{s_i^*}$ after either acquiring the input $v^*$, receiving $\msgpair{\YOURS}{s_i^*}$ from $t + 1$ parties, or setting $y_i \gets v^*$. As a consequence, every honest $P_i$ eventually receives $n - t$ $\MINE$ symbols that are correct w.r.t.\ $v^*$, stores these symbols in $z_1, \dots, z_n$, computes $y \gets \Dec_{n - 2t}(z_1, \dots, z_n)$ on line \ref{line:l12} where $y = v^*$ due to $\Dec_{n - 2t}$'s ability to tolerate $t$ missing or incorrect symbols, verifies that $y \neq \bot$ on line \ref{line:l13} and that at least $n - t$ of $z_1, \dots, z_n$ are correct w.r.t.\ $y$ on line \ref{line:l15}, sets $y_i$ to $v^*$ on line \ref{line:l16} and finally sends everyone $\YOURS$ messages on line \ref{line:l18}. Finally, the honest parties all sending each other $\YOURS$ messages ensures that every honest party receives $2t + 1$ $\YOURS$ messages, and thus becomes able to output $y_i = v^*$ and terminate. 

    Both liveness (with termination) and totality follow from the fact that if $t + 1$ honest parties send everyone $\YOURS$ messages, then everyone terminates. Liveness follows from this fact because if $t + 1$ honest parties acquire the input $v^*$, then these parties send everyone $\YOURS$ messages. Meanwhile, totality follows from this fact because an honest party can only terminate after receiving $2t + 1$ $\YOURS$ messages, at least $t + 1$ of which are honest parties who have sent everyone $\YOURS$ messages.
\end{proof}

\paragraph*{Complexity of \textsf{REC}} If there is some $v^*$ where $\Enc_{n - 2t}(v^*) = (s_1^*, \dots, s_n^*)$ such that $v^*$ is the only input the honest parties can acquire, then every honest party terminates $\rec$ in a constant number of rounds after either $t + 1$ honest parties acquire the input $v^*$ or some honest party terminates. Meanwhile, even if the honest parties can acquire different inputs, the message and communication complexities or $\rec$ are respectively $\Theta(n^2)$ and $\Theta(\ell n + n^2)$ due to each party sending everyone one $\MINE$ symbol and one $\YOURS$ symbol, where each symbol is of size $\Theta\bigl(\frac{\ell}{n - 2t} + 1\bigr) = \Theta\bigl(\frac{\ell}{n} + 1\bigr)$.

\section{Details on Almost-Universal Hashing} \label{almostunivappendix}

For any message length $\ell \geq 1$ and any key/hash length $\kappa \geq 1$ which divides $\ell$ such that $2^{\kappa} \geq \frac{\ell}{\kappa}$, one can obtain an $\big(\frac{\ell \cdot 2^{-\kappa}}{\kappa}\big)$-almost universal keyed hash function $h$ as follows \cite{fh06, cw79}: Split the hashed message $m \in \{0, 1\}^\ell$ into $\frac{\ell}{\kappa}$ equal-sized symbols $s_0, \dots, s_{\frac{\ell}{\kappa} - 1}$ in the Galois Field $GF(2^\kappa)$, interpolate the unique polynomial $p_m$ of degree at most $\frac{\ell}{\kappa} - 1$ such that $p_m(j) = s_j$ for all $j \in \{0, \dots, \frac{\ell}{\kappa} - 1\}$, and let $h(k, m) = p_m(k)$. For any distinct $a, b \in \{0, 1\}^\ell$, the interpolated polynomials $p_a$ and $p_b$ coincide at at most $\frac{\ell}{\kappa}$ points; so, for a random key $k \in \{0, \dots, 2^\kappa - 1\}$, we have $\Pr[h(k, a) = h(k, b)] \leq \frac{\ell \cdot 2^{-\kappa}}{\kappa}$.

In Section \ref{statsec}, we have a security parameter $\lambda$, and we want a $\big(\frac{2^{-\lambda}}{n^2}\big)$-almost universal keyed hash function $h$. For this, we set $\kappa = \ceil{\lambda + \log_2(\ell n^2) + 1}$. Suppose $\ell \geq \kappa$, as $\ell < \kappa$ would allow the parties to exchange messages instead of hashes without any increase in complexity. Since we need $\kappa$ to divide $\ell$, before hashing messages of length $\ell$, we pad them with up to $\kappa - 1$ zeroes into messages of length $\ell' = \kappa \cdot \ceil[\big]{\frac{\ell}{\kappa}} < \kappa \cdot \frac{2\ell}{\kappa} = 2\ell$. So, we ensure that distinct padded $\ell'$-bit messages $a$ and $b$ collide on a random key with probability at most $\frac{\ell' \cdot 2^{-\kappa}}{\kappa} < 2\ell \cdot 2^{-\lambda - \log_2(\ell n^2) - 1} = \frac{2\ell}{2^{\lambda + 1}\ell n^2} = \frac{2^{-\lambda}}{n^2}$.

\section{Proof of Lemma \ref{graphlemma}} \label{kwaappendix}

\graphlemma*

\begin{proof}
    Let $G = (V, E), n, t, a, C, D$ all be as described in the statement above. For all $(x, y) \in V^2$, let $E_{x, y} = 1$ if there exists an edge in $G$ between the (possibly identical) vertices $x$ and $y$, and let $E_{x, y} = 0$ otherwise. Since every $x \in C$ has a closed neighborhood in $G$ of size at least $n - 3t$, we have $\sum_{x \in C}\sum_{y \in V}E_{x, y} \geq \sum_{x \in C}(n - 3t) = (n - 3t)^2$. Meanwhile, the vertices in $D$ have at most $|C| = n - 3t$ neighbors in $C$, which implies $\sum_{x \in D}\sum_{y \in C}E_{x, y} \leq \sum_{x \in D}(n - 3t) = |D|(n - 3t)$, and the vertices outside $D$ must (by definition since they are not in $D$) have less than $a(n - 3t)$ neighbors in $C$, which implies $\sum_{x \in V \setminus D}\sum_{y \in C}E_{x, y} \leq \sum_{x \in V \setminus D}a(n - 3t) = (n - t - |D|)a(n - 3t)$, with equality if and only if $V \setminus D = \varnothing$. Therefore, we have $(n - 3t)^2 \leq \sum_{x \in C}\sum_{y \in V}E_{x, y} = \sum_{y \in V}\sum_{x \in C}E_{y, x} = (\sum_{y \in D}\sum_{x \in C}E_{y, x}) + (\sum_{y \in V \setminus D}\sum_{x \in C}E_{y, x}) \leq |D|(n - 3t) + (n - t - |D|)a(n - 3t)$. Finally, simplifying the inequality $(n - 3t)^2 \leq |D|(n - 3t) + (n - t - |D|)a(n - 3t)$ by dividing both of its sides by $n - 3t$ and solving it for $|D|$ gives us $|D| \geq n - (3 + \frac{2a}{1 - a})t$.
\end{proof}

\end{document}